%% file: paper.tex
\pdfoutput=1

\documentclass{TOOLS/sigplanconf}
\usepackage{amsfonts,amsmath}

\newenvironment{proof}[1][Proof]{\begin{trivlist}
\item[\hskip \labelsep {\bfseries #1}]}{\end{trivlist}}
\input TOOLS/chheader13.tex

\begin{document}

\titlebanner{}      
\preprintfooter{}   

\title{ 
  Arithmetic Algorithms for
  Hereditarily Binary Natural Numbers
}
\subtitle{}
           
\authorinfo{Paul Tarau}
   {Department of Computer Science and Engineering\\
   University of North Texas}
   {\em tarau@cs.unt.edu}

\maketitle

\begin{abstract}

We study some essential arithmetic properties
of a new tree-based number representation,
{\em hereditarily binary numbers},
defined by applying recursively
run-length encoding of 
bijective base-2 digits.
  
Our representation expresses
giant numbers like
the largest known prime number and
its related perfect number as well as
the largest known
Woodall, Cullen, Proth, Sophie Germain and twin primes 
as trees of small sizes.

More importantly, our number representation 
supports novel algorithms that, in the best case,
collapse the complexity of various computations
by super-exponential factors and in the worse
case are within a constant factor of their
traditional counterparts.

As a result, it opens the door to a new world,
where arithmetic
operations are limited by the structural complexity
of their operands, rather than their bitsizes.



\end{abstract}

\section{Introduction}
Number representations have evolved over time from the unary ``cave man'' representation
where one scratch on the wall represented a unit, to the base-n (and in particular base-2) number system, with the remarkable benefit of 
a logarithmic representation size. 
Over the last 1000 years,
this base-$n$ representation has proved to be unusually resilient,
partly because all practical computations could be performed
with reasonable efficiency within the notation.

However, when thinking ahead for the next 1000 years,
computations with very large numbers 
are likely to become more and more ``natural'', even if 
for now, they are mostly driven 
by purely theoretical interests in fields like
number theory, computability or multiverse cosmology.
Arguably, more practical needs of present and future cryptographic systems 
might also justify devising alternative numbering systems
with higher limits on the size of the numbers with which we
can perform tractable computations.

While {\em notations} like Knuth's ``up-arrow'' \cite{knuthUp} or 
tetration 
are useful in describing very large numbers, they do not provide the ability to actually {\em compute} with them -- as, for instance, addition or multiplication with a natural number results in a number that
cannot be expressed with the notation anymore.
More exotic notations like Conway's surreal numbers \cite{surreal} 
involve uncountable cardinalities (they
contain real numbers as a subset) and are more useful for modeling
game-theoretical algorithms rather than common arithmetic computations.

The novel contribution of this paper is a tree-based numbering
system that {\em allows computations} with numbers comparable in size
with Knuth's ``arrow-up'' notation.
Moreover, these computations have a worse case complexity
that is comparable with the traditional  binary numbers,
while their best case complexity outperforms binary numbers by an arbitrary
tower of exponents factor. Simple operations like successor, 
multiplication by 2, exponent of 2 are practically constant time and
a number of other operations benefit from
significant complexity reductions.

For the curious reader, it is basically
a {\em hereditary number system} \cite{goodstein},
based on recursively applied {\em run-length} 
compression of a special (bijective) binary digit notation.

A concept of structural complexity is introduced, based on the
size of our tree representations and it is shown that several
``record holder'' large numbers like Mersenne, Cullen, Woodall and
Proth primes have unusually small structural complexities.

We have adopted
a {\em literate programming} style, i.e. the
code contained in the paper
forms a self-contained Haskell module (tested with ghc 7.6.3),
also available as a separate file at
\url{http://logic.cse.unt.edu/tarau/research/2013/hbin.hs} .
Alternatively, a Scala package implementing the same tree-based computations
is available from \url{http://code.google.com/p/giant-numbers/}.
We hope that this will encourage the reader to
experiment interactively and validate the
technical correctness of our claims. The {\bf Appendix}
contains a quick overview of the subset of Haskell we are
using as our executable function notation.

The paper is organized as follows.
Section \ref{bijbin} gives some background on bijective base-2 numbers
and iterated
function applications.
Section \ref{herbin} introduces hereditarily binary numbers.
Section \ref{succ} describes practically constant time successor and predecessor operations
on tree-represented numbers.
Section \ref{emu} shows an emulation of bijective base-2 with hereditarily binary numbers
and section \ref{arith} describes
novel algorithms for arithmetic operations taking advantage of our number
representation.
Section \ref{stru} defines a concept of structural complexity and studies
best and worse cases.
Section \ref{theo} describes efficient tree-representations 
of some important number-theoretical entities like Mersenne, Fermat,
Proth, Woodall primes. 
Section \ref{related} discusses related work.
Section \ref{concl} concludes the paper and discusses future work.

\section{Natural numbers as iterated function applications} \label{bijbin}

Natural numbers can be seen as represented by iterated
applications of the functions $o(x)=2x+1$ and $i(x)=2x+2$
corresponding the so called
{\em bijective base-2} representation \cite{sac12} 
together with the convention
that 0 is represented as the empty sequence. 
As each $n \in \N$ can be seen as a unique composition of these functions
we can make this precise as follows:
\begin{df}
We call bijective base-2 representation of $n \in \N$ the 
unique sequence of applications of functions $o$ and $i$ 
to  $\epsilon$ that evaluates to $n$.
\end{df}
With this representation,
and denoting the empty sequence $\epsilon$, one obtains 
$0=\epsilon, 
1=o~\epsilon,
2=i~\epsilon,
3=o (o~\epsilon),
4=i (o~\epsilon),
5=o (i~\epsilon)$ etc. 
and the following holds:
\begin{equation} \label{isucco}
  i(x)=o(x)+1
\end{equation}

\subsection{Properties of the iterated functions $o^n$ and $i^n$}

\begin{prop}\label{fastexp}
Let $f^n$ denote application of function $f$ $n$ times. Let $o(x)=2x+1$ and $i(x)=2x+2$, $s(x)=x+1$ and $s'(x)=x-1$. Then $k>0 \Rightarrow s(o^n(s'(k))=k2^n$ and $k>1 \Rightarrow s(s(i^n(s'(s'(k))))=k2^n$. In particular, $s(o^n(0))=2^n$ and $s(s(i^n(0)))=2^{n+1}$.
\end{prop}
\begin{proof}
By induction.
Observe that for $n=0,k>0, s(o^0(s'(k))=k2^0$ because $s(s'(k)))=k$. 
Suppose that $P(n): k>0 \Rightarrow s(o^n(s'(k)))=k2^n$ holds. 
Then, assuming $k>0$, P(n+1) follows, given that
$s(o^{n+1}(s'(k)))=s(o^n(o(s'(k))))=s(o^n(s'(2k)))= 2k2^n= k2^{n+1}$.
Similarly, the second part of the proposition also follows by induction on $n$.
\end{proof}
The underlying arithmetic identities are:
\begin{equation}
k>0 \Rightarrow 1+o^n(k-1)=2^nk \label{on}
\end{equation}

\begin{equation}\label{in}
k>1 \Rightarrow 2+i^n(k-2)=2^nk
\end{equation}
from where one can deduce 
\begin{equation}\label{onk}
o^n(k)=2^n(k+1)-1
\end{equation}
\begin{equation}\label{ink}
i^n(k)=2^n(k+2)-2
\end{equation}
and in particular
\begin{equation}
o^n(0)=2^n-1
\end{equation}
\begin{equation}
i^n(0)=2^{n+1}-2
\end{equation}
Also, one can directly relate $o^k$ and $i^k$
\begin{equation}
o^n(k)+2^n=i^n(k)+1
\end{equation}
\begin{equation}
i^n(k)=o^n(k)+o^n(0)
\end{equation}
\begin{equation}
o^{n}(k+1)=i^n(k)+1
\end{equation}

\subsection{The iterated functions $o^n$, $i^n$ and their conjugacy results}
Results from the theory of {\em iterated functions} apply to our operations.
The following proposition is proven in \cite{iterfun}:
\begin{prop}
If $f(x)=ax+b$ with $a \neq 1$, 
let $x_0$ be such that $f(x_0)=x_0$ 
i.e. $x_0={b \over {1-a}}$. Then $f^n(x)=x_0+(x-x_0)a^n$.
\end{prop}

For $a=2,b=1$ and respectively $a=2,b=2$ this provides an 
alternative proof for  proposition \ref{fastexp}.

A few properties similar to {\em topological conjugation} apply to our
functions
\begin{df}
We call two functions $f,g$ conjugates through $h$ if $h$ is a bijection such that 
$g=h^{-1} \circ f \circ h$, where $\circ$ denotes function composition.
\end{df}
\begin{prop}\label{conjn}
If $f,g$ are conjugates through $h$ then $f^n$ and $g^n$ are too, i.e. 
$g^n=h^{-1} \circ {f^n} \circ h$.
\end{prop}
\begin{proof}
By induction, using the fact that $h^{-1} \circ h = \lambda x.x = h \circ h^{-1} $.
\end{proof}

\begin{prop}
$o^n$ and $i^n$ are conjugates with respect to $s$ and $s'$, i.e. the following 2 identities hold:
\end{prop}

\begin{equation}
o^{n}(k)=s(i^n(s'(k)))
\end{equation}

\begin{equation}
i^n(k)=s'(o^{n}(s(k)))
\end{equation}
\begin{proof}
An immediate consequence of $i=s \circ o$ (by \ref{isucco}) and Prop. \ref{conjn}.
\end{proof}

Note also that proposition \ref{fastexp} can be seen as stating that $o^n$ is the conjugate of the leftshift operation $l(n,x)=2^{n}x$ through $s(x)=x+1$ (eq. \ref{on}) 
and so is $i^n$ through $s \circ s$ (eq. \ref{in}).

The following equations relate successor and predecessor to the iterated applications of $o$ and $i$:

\begin{equation}\label{sonk}
s(o^{n}(k))=i(o^{s'(n)}(k))
\end{equation}

\begin{equation}
s(i^n(k))=o^{n}(s(k))
\end{equation}

\begin{equation}
s'(o^n(k))=i^n(s'(k))
\end{equation}

\begin{equation}\label{sink}
s'(i^{n}(k))=o(i^{s'(n)}(k))
\end{equation}

By setting $k=2m+1$ in eq. \ref{on} we obtain:
\begin{equation}
1+o^n(2m)=2^n(2m+1)
\end{equation}
As the right side of this equation expresses
a bijection between $\N \times \N$ and $\N^{+}$, so does the left side, i.e. the function $c(m,n)=1+o^n(2m)$ maps pairs (m,n) to unique values in $\N^{+}$.

Similarly, by setting $k=2m+1$ in eq. \ref{in} we obtain:
\begin{equation}
2+i^n(2m-1)=2^n(2m+1)
\end{equation}

\section{Hereditarily binary numbers} \label{herbin}

\subsection{Hereditary Number Systems}
Let us observe that conventional number systems, as well as the bijective base-2 numeration system described so far, represent 
blocks of 0 and 1 digits somewhat naively - one 
digit for each element of the block.
Alternatively, one might think that counting them 
and representing the resulting
counters as {\em binary numbers} would be also possible. 
But then, the same principle could be  applied recursively.
So instead of representing each block of 0 or 1 digits by as many
symbols as the size of the block -- essentially a {\em unary} representation --
one could also encode the number of elements in such a block using
a {\em binary} representation.

This brings us to the idea of hereditary number systems. At
our best knowledge the first instance of such a system is used in
\cite{goodstein}, by iterating the polynomial base-n notation
to the exponents used in the notation.
We next explore a hereditary number representation
that implements the simple idea of representing
the number of contiguous 0 or 1 digits in a number,
as bijective base-2
numbers, recursively.


\subsection{Hereditarily binary numbers as a data type}

First, we define a data type for our tree represented natural numbers,
that we call {\em hereditarily binary numbers} to emphasize that
{\em binary} rather than {\em unary}
encoding is recursively used in their representation.
\begin{df}
The data type $\T$ of the set of hereditarily binary numbers 
is defined by the Haskell declaration:
\end{df}
\begin{code}
data T = E | V T [T]  | W T [T] deriving (Eq,Show,Read)
\end{code}
that automatically derives the equality relation ``{\tt ==}'', as well
as reading and string representation. For shortness,
We will call the members of type $\T$ {\em terms}.
The intuition behind the disjoint union type $\T$ is the following:
\begin{itemize}
\item The term {\tt E} (empty leaf) corresponds to zero 
\item the term {\tt V x xs} counts the number {\tt x+1} of {\tt o} applications followed by an {\em alternation}
of similar counts of {\tt i} and {\tt o} applications

\item the term {\tt W x xs} counts the number {\tt x+1} of {\tt i} applications followed by an {\em alternation}
of similar counts of {\tt o} and {\tt i} applications

\item the same principle is applied recursively for the counters, until the empty sequence is reached
\end{itemize}
One can see this process as run-length compressed bijective base-2 numbers, represented as trees
with either empty leaves or at least one branch, after applying the encoding recursively.

These trees can be specified in the proof assistant {\em Coq} \cite{Coq:manual} as the type {\tt T}:
\begin{codex}
Require Import List.
Inductive T : Type :=
| E : T
| V : T -> list T -> T
| W : T -> list T -> T.
\end{codex}
which automatically generates the induction principle:
\begin{codex}
Coq < Check T_ind.
hbNat_ind
     : forall P : T -> Prop,
       P E ->
       (forall h : T, P h -> 
          forall l : list T, P (V h l)) ->
       (forall h : T, P h -> 
          forall l : list T, P (W h l)) ->
       forall h : T, P h
\end{codex}

\begin{df}
The function $n:\T \to \N$  shown in equation {\bf \ref{deft}} defines the unique natural number associated to a term of type $\T$.
\begin{figure*}
\begin{equation}\label{deft}
n(t)=
\begin{cases}
0&  \text{if $t=~${\tt E}},\\
2^{n(x)+1}-1&  \text{if $t=~${\tt V x []}},\\

{(n(u)+1)}2^{n(x)+1}-1&  \text{if $t=~${\tt V x (y:xs)} and $u=~${\tt W y xs}},\\

2^{n(x)+2}-2&  \text{if $t=~${\tt W x []}},\\

{(n(u)+2)}2^{n(x)+1}-1&  \text{if $t=~${\tt W x (y:xs)} and $u=~${\tt V y xs}}.
\end{cases}
\end{equation}
\end{figure*}
\end{df}
For instance, the computation of n(W (V E []) [E,E,E])$~=42$ expands to
$({({({2^{{0}+1}-1}+2)2^{{0} + 1}-2}+1)2^{{0} + 1}-1}+2)2^{{2^{{0}+1}-1} + 1}-2$.
The Haskell equivalent of equation (\ref{deft}) is:
\begin{code}
n E = 0
n (V x []) = 2^(n x +1)-1
n (V x (y:xs)) = (n u+1)*2^(n x + 1)-1 where u = W y xs
n (W x []) = 2^(n x+2)-2
n (W x (y:xs)) = (n u+2)*2^(n x + 1)-2 where u = V y xs  
\end{code}
The following example illustrates the values associated with
the first few natural numbers.
\begin{codex}
0 = n E
1 = n (V E [])
2 = n (W E [])
3 = n (V (V E []) [])
4 = n (W E [E])
5 = n (V E [E])
\end{codex}
Note that a term of the form {\tt V x xs} represents an
odd number $\in \N^+$ and a term of the form {\tt W x xs} represents an
even number $\in \N^+$. The following holds:
\begin{prop}
$n:\T \to \N$ is a bijection, i.e., each term canonically represents the
corresponding natural number.
\end{prop}
\begin{proof}
It follows from the equations \ref{onk}, \ref{ink} by 
replacing the power of 2 functions with the
corresponding iterated applications of
$o$ and $i$.
\end{proof}

\section{Successor ({\tt s}) and predecessor ({\tt s'})} \label{succ}

We will now specify successor and predecessor on data type $\T$
through two mutually recursive functions defined in the
proof assistant {\em Coq}
as

\begin{codex}
Fixpoint s (t:T) :=
match t with
| E => V E nil
| V E nil =>  W E nil
| V E (x::xs) =>  W (s x) xs 
| V z xs => W E (s' z :: xs)
| W z nil => V (s z) nil
| W z (E::nil) => V z (E::nil)
| W z (E::y::ys) => V z (s y::ys)
| W z (x::xs) => V z (E::s' x::xs)
end 
with 
s' (t:T) : T :=
match t with
| V E nil => E
| V z nil => W (s' z) nil
| V z (E::nil) =>  W z (E::nil)
| V z (E::x::xs) =>  W z (s x::xs)  
| V z (x::xs) =>  W z (E::s' x::xs) 
| W E nil => V E nil
| W E (x::xs) => V (s x) xs 
| W z xs => V E (s' z::xs)
| E => E (* Coq wants t total on T *)
end.
\end{codex}
Note that our definitions are conforming with {\em Coq's}
requirement for automatically guaranteeing termination, i.e.
they use only induction on the structure of the terms.
Once accepting the definition, {\em Coq} allows extraction
of equivalent Haskell code, that, in a more human readable form
looks as follows:
\begin{code}
s E = V E []
s (V E []) =  W E []
s (V E (x:xs)) =  W (s x) xs  
s (V z xs) = W E (s' z : xs)
s (W z [])  = V (s z) []
s (W z [E]) = V z [E]
s (W z (E:y:ys)) = V z (s y:ys)
s (W z (x:xs)) = V z (E:s' x:xs)
\end{code}
\begin{code}
s' (V E []) = E
s' (V z []) = W (s' z) []
s' (V z [E]) =  W z [E]
s' (V z (E:x:xs)) =  W z (s x:xs)  
s' (V z (x:xs)) =  W z (E:s' x:xs)  
s' (W E []) = V E []
s' (W E (x:xs)) = V (s x) xs 
s' (W z xs) = V E (s' z:xs)
\end{code}

The following holds:
\begin{prop}
Denote $\T^+=\T-\{\text{\tt E}\}$.
The functions $s:\T \to \T^+$ and $s':\T^+ \to \T$ are inverses.
\end{prop}
\begin{proof}
It follows by structural induction after observing that 
patterns for {\tt V} in {\tt s} correspond one by one to patterns for
{\tt W} in {\tt s'} and vice versa.
\end{proof}

More generally, it can be proved by structural induction that Peano's axioms hold
and as a result $<\T,E,s>$ is a Peano algebra.

Note also that calls to {\tt s,s'} in {\tt s} or {\tt s'} happen on terms that are
(roughly) logarithmic in the bitsize of their operands.  One can therefore 
assume that their complexity, computed by an {\em iterated logarithm},
is practically constant.

\section{Emulating the bijective base-2 operations {\tt o}, $i$} \label{emu}

To be of any practical interest, we will need to ensure that
our data type $\T$ emulates also binary arithmetic. We will
first show that it does, and next we will show that
on a number of operations like exponent of 2 or multiplication
by an exponent of 2, it significantly lowers complexity.

Intuitively the first step should be easy, as we need
to express single applications or ``un-applications'' of
{\tt o} and {\tt i} in terms of their iterates
encapsulated in the {\tt V} and {\tt W} terms. 

First we emulate single applications of {\tt o} and {\tt i} seen
as virtual ``constructors'' on type \verb~data B = Zero | O B | I B~.
\begin{code}
o E = V E []
o (V x xs) = V (s x) xs
o (W x xs) = V E (x:xs)

i E = W E []
i (V x xs) = W E (x:xs)
i (W x xs) = W (s x) xs
\end{code}
Next we emulate the corresponding ``destructors'' that can be seen
as ``un-applying'' a single instance of {\tt o} or {\tt i}.
\begin{code}  
o' (V E []) = E
o' (V E (x:xs)) = W x xs
o' (V x xs) = V (s' x) xs

i' (W E []) = E
i' (W E (x:xs)) = V x xs
i' (W x xs) = W (s' x) xs
\end{code}
Finally the ``recognizers'' {\tt o\_} (corresponding to {\em odd} numbers)
and {\tt i\_} (corresponding to {\em even} numbers)
simply detect {\tt V} and {\tt W} corresponding
to {\tt o} (and respectively {\tt i}) being the last function applied. 
\begin{code}
o_ (V _ _ ) = True
o_ _ = False

i_ (W _ _ ) = True
i_ _ = False
\end{code}
Note that each of the functions {\tt o,o'} and 
{\tt i,i'} call {\tt s} and {\tt s'} on a term that is (roughly)
logarithmically smaller. It follows that
\begin{prop}
Assuming {\tt s,s'} constant time, {\tt o,o',o,i'} are also constant time.
\end{prop}

\begin{df}
The function $t:\N \to \T$ defines the unique tree of type $\T$ associated to a natural number as follows:
\begin{equation}
t(x)=
\begin{cases}
{\tt E}&  \text{if $x=~${\tt 0}},\\
\text{\tt o}(t({{x-1}\over 2})) &  \text{if $x>0$ and $x$ is odd},\\
\text{\tt i}(t({x \over 2}-1)) &  \text{if $x>0$ and $x$ is even},\\
\end{cases}
\end{equation}
\end{df}

We can now define the corresponding Haskell function {\tt t}:$\T \to \N$ that converts from trees to
natural numbers. 

Note that {\tt pred x=x-1} and {\tt div} is integer division.
\begin{code}
t 0 = E
t x | x>0 && odd x = o(t (div (pred x) 2))
t x | x>0 && even x = i(t (pred (div x 2)))
\end{code}
The following holds: 

\begin{prop}
Let {\tt id} denote $\lambda x.x$ and $\circ$ function composition. Then, on their respective domains
\begin{equation}
t \circ n = id,~~
n \circ t = id
\end{equation}
\end{prop}
\begin{proof}
By induction, using the arithmetic formulas defining the two functions.
\end{proof}

\section{Arithmetic operations} \label{arith}
We will now describe algorithms for
basic arithmetic operations that take advantage of
our number representation.

\subsection{A few low complexity operations}
Doubling a number {\tt db}
and reversing the {\tt db} operation ({\tt hf}) are
quite simple, once one remembers that the arithmetic
equivalent of function {\tt o} is $o(x)=2x+1$.
\begin{code}
db = s' . o
hf = o' . s 
\end{code}

Note that efficient implementations follow directly from
our number theoretic observations in section \ref{bijbin}. 

For instance, as a consequence
of proposition \ref{fastexp},
the operation
{\tt exp2}, computing an exponent of $2$ ,
has the following simple definition in terms of
{\tt s} and {\tt s'}.
\begin{code}
exp2 E = V E []
exp2 x = s (V (s' x) [])
\end{code}

\begin{prop}
Assuming {\tt s,s'} constant time, {\tt db,hf} and {\tt exp2} are also constant time.
\end{prop}
\begin{proof}
It follows by observing that only 2 calls to {\tt s,s',o,o'} are
made. 
\end{proof}

\subsection{Simple addition and subtraction algorithms}

A simple addition algorithm ({\tt add}) proceeds by recursing through 
our emulated bijective base-2
operations {\tt o} and {\tt i}. 
\begin{code}
simpleAdd E y = y
simpleAdd x E = x
simpleAdd x y | o_ x && o_ y = 
  i (simpleAdd (o' x) (o' y))
simpleAdd x y | o_ x && i_ y = 
  o (s (simpleAdd (o' x) (i' y)))
simpleAdd x y | i_ x && o_ y = 
  o (s (simpleAdd (i' x) (o' y)))
simpleAdd x y | i_ x && i_ y = 
  i (s (simpleAdd (i' x) (i' y)))
\end{code}
Subtraction is similar.
\begin{code}
simpleSub x E= x
simpleSub y x | o_ y && o_ x = 
  s' (o (simpleSub (o' y) (o' x))) 
simpleSub y x | o_ y && i_ x = 
  s' (s' (o (simpleSub (o' y) (i' x))))
simpleSub y x | i_ y && o_ x = 
  o (simpleSub (i' y) (o' x))  
simpleSub y x | i_ y && i_ x = 
  s' (o (simpleSub (i' y) (i' x))) 
\end{code}

The following holds:
\begin{prop}
Assuming that {\tt s,s} are constant time, the complexity of addition and subtraction
is proportional to the bitsize of
the smallest of their operands.
\end{prop}
\begin{proof}
The algorithms advance through both 
their operands with one {\tt o',i'} operation at each call.
Therefore the case when one operand is {\tt E} is reached
as soon as the shortest operand is processed.
\end{proof}

Note that these simple algorithms do not take advantage of
possibly large blocks of $o$ and $i$ operations that
could significantly lower complexity on large numbers
with such a ``regular'' structure.

We will derive next versions of these algorithms {\em favoring
terms with large contiguous blocks of $o^n$ and $i$ applications},
on which they will lower complexity to depend on the {\em number of
blocks rather than the total number of $o$ and $i$ applications
forming the blocks}.

Given the recursive, self-similar structure of our trees, as
the algorithms mimic the data structures they operate on,
we will have to work with a chain of {\em mutually recursive}
functions. As our focus is to take advantage of large
contiguous blocks of $o^n$ and $i^m$ applications, 
the algorithms are in uncharted territory and as a result
somewhat more intricate than their traditional counterparts.

\subsection{Reduced complexity addition and subtraction}

We derive more efficient addition and subtraction
operations similar
to {\tt s} and {\tt s'} that {\em work on one
run-length encoded block at a time}, rather than by
individual {\tt o} and {\tt i} steps.

We  first
define the functions {\tt otimes} corresponding to $o^n(k)$ and {\tt itimes} corresponding to $i^n(k)$.
\begin{code}
otimes E y = y
otimes n E = V (s' n) []
otimes n (V y ys) = V (add n y) ys
otimes n (W y ys) = V (s' n) (y:ys)
\end{code}
\begin{code}
itimes E y = y
itimes n E = W (s' n) []
itimes n (W y ys) = W (add n y) ys
itimes n (V y ys) = W (s' n) (y:ys)    
\end{code}
They are part of a chain of mutually recursive functions as
they are already referring to the
{\tt add} function, to be implemented later.
Note also that instead of naively iterating, they implement a
more efficient algorithm, working
 ``one block at a time''. When detecting that its argument counts
a number of applications of {\tt o}, {\tt otimes} just increments that count.
On the other hand, when the last function applied was {\tt i}, {\tt otimes} simply
inserts a new count for {\tt o} operations. A similar process corresponds to 
{\tt itimes}. As a result, performance is (roughly) logarithmic rather than linear
in terms of the bitsize of argument {\tt n}. We will also use this property for
implementing a low complexity multiplication by exponent of 2 operation.

We will state a number of arithmetic identities on $\N$
involving iterated applications of $o$ and $i$.
\begin{prop} \label{addeqs}
The following hold:
\begin{equation} \label{oplus}
o^k(x) + o^k(y) = i^k(x+y) 
\end{equation}
\begin{equation} \label{oiplus}
o^k(x) + i^k(y) = i^k(x)+o^k(y) = i^k(x+y+1)-1 
\end{equation}
\begin{equation} \label{iplus}
i^k(x) + i^k(y) = i^k(x+y+2)-2 
\end{equation}
\end{prop}
\begin{proof}
By (\ref{onk}) and (\ref{ink}), we substitute the $2^k$-based equivalents of $o^k$ and $i^k$,
then observe that the same reduced forms appear on both sides.
\end{proof}

The corresponding Haskell code is:
\begin{code}
oplus k x y =  itimes k (add x y) 
   
oiplus k x y = s' (itimes k (s (add x y)))    
   
iplus k x y = s' (s' (itimes k (s (s (add x y)))))
\end{code}
Note the use of {\tt add} that we will define later as
part of a chain of mutually recursive function calls,
that together will provide an implementation of
the intuitively simple idea: {\em they work on one run-length 
encoded block at a time}. {\em ``Efficiency'', in what follows,
will be, therefore, conditional to numbers having comparatively 
few such blocks}.

The corresponding identities for subtraction are:
\begin{prop} \label{subeqs}
\begin{equation}\label{ominus}
x > y ~\Rightarrow~ o^k(x) - o^k(y) = o^k(x-y-1)+1 
\end{equation}
\begin{equation}\label{oiminus}
x>y+1 ~\Rightarrow~ o^k(x) - i^k(y) = o^k(x-y-2)+2
\end{equation}
\begin{equation}\label{iominus}
x \geq y ~\Rightarrow~ i^k(x) - o^k(y) = o^k(x-y)
\end{equation}
\begin{equation}\label{iminus}
x > y ~\Rightarrow~ i^k(x) - i^k(y) = o^k(x-y-1)+1 
\end{equation}
\end{prop}
\begin{proof}
By (\ref{onk}) and (\ref{ink}), we substitute the $2^k$-based equivalents of $o^k$ and $i^k$, and observe that the same reduced forms appear on both sides. Note that special cases
are handled separately to ensure that subtraction is defined.
\end{proof}
The Haskell code, also covering the special cases, is:
\begin{code}
ominus _ x y | x == y = E
ominus k x y = s (otimes k (s' (sub x y))) 

iminus _ x y | x == y = E   
iminus k x y =  s (otimes k (s' (sub x y)))

oiminus k x y | x==s y = s E
oiminus k x y | x == s (s y) = s (exp2 k)
oiminus k x y =  s (s (otimes k (s' (s' (sub x y)))))   

iominus k x y = otimes k (sub x y)
\end{code}
Note the reference to {\tt sub}, to be defined later, which
is also part of the mutually recursive chain of operations.

The next two functions extract the iterated applications of
$o^n$ and respectively $i^n$ from {\tt V} and {\tt W} terms:
\begin{code}
osplit (V x []) = (x,E )
osplit (V x (y:xs)) = (x,W y xs)

isplit (W x []) = (x,E )
isplit (W x (y:xs)) = (x,V y xs)
\end{code}
We are now ready for defining addition. The base cases are the same as for {\tt simpleAdd}:
\begin{code}
add E y = y
add x E = x
\end{code}
In the case when both terms represent odd numbers, we apply the identity (\ref{oplus}),
after extracting the iterated applications of $o$ as {\tt a} and {\tt b} with the function {\tt osplit}.
\begin{code}
add x y |o_ x && o_ y = f (cmp a b) where
  (a,as) = osplit x
  (b,bs) = osplit y
  f EQ = oplus (s a) as bs
  f GT = oplus (s b) (otimes (sub a b) as) bs
  f LT = oplus (s a) as (otimes (sub b a) bs)
\end{code}
In the case when the first term is odd and the second even, we apply the identity (\ref{oiplus}),
after extracting the iterated application of $o$ and $i$ as {\tt a} and {\tt b}.
\begin{code}  
add x y |o_ x && i_ y = f (cmp a b) where
  (a,as) = osplit x
  (b,bs) = isplit y
  f EQ = oiplus (s a) as bs
  f GT = oiplus (s b) (otimes (sub a b) as) bs
  f LT = oiplus (s a) as (itimes (sub b a) bs)  
\end{code}
In the case when the first term is even and the second odd, we apply the identity (\ref{oiplus}),
after extracting the iterated applications of $i$ and $o$ as, respectively, {\tt a} and {\tt b}.
\begin{code}    
add x y |i_ x && o_ y = f (cmp a b) where
  (a,as) = isplit x
  (b,bs) = osplit y
  f EQ = oiplus (s a) as bs
  f GT = oiplus (s b) (itimes (sub a b) as) bs
  f LT = oiplus (s a) as (otimes (sub b a) bs)    
\end{code}
In the case when both terms represent even numbers, we apply the identity (\ref{iplus}),
after extracting the iterated application of $i$ as {\tt a} and {\tt b}.
\begin{code}      
add x y |i_ x && i_ y = f (cmp a b) where
  (a,as) = isplit x
  (b,bs) = isplit y
  f EQ = iplus (s a) as bs
  f GT = iplus (s b) (itimes (sub a b) as) bs
  f LT = iplus (s a) as (itimes (sub b a) bs)  
\end{code}
Note the presence of the comparison operation {\tt cmp}, to be defined later, also
part of our chain of mutually recursive operations.
Note also the local function {\tt f} that in each case ensures that a block of the same
size is extracted, depending on which of the two operands {\tt a} or {\tt b} is larger.
The code for the subtraction function {\tt sub} is similar:
\begin{code}
sub x E = x
sub x y | o_ x && o_ y = f (cmp a b) where
  (a,as) = osplit x
  (b,bs) = osplit y
  f EQ = ominus (s a) as bs
  f GT = ominus (s b) (otimes (sub a b) as) bs
  f LT = ominus (s a) as (otimes (sub b a) bs)
\end{code}
In the case when both terms represent odd numbers, we apply the identity (\ref{ominus}),
after extracting the iterated applications of $o$ as {\tt a} and {\tt b}.
For the other cases, we use, respectively, the identities \ref{oiminus}, \ref{iominus} and
\ref{iminus}:
\begin{code}  
sub x y |o_ x && i_ y = f (cmp a b) where
  (a,as) = osplit x
  (b,bs) = isplit y
  f EQ = oiminus (s a) as bs
  f GT = oiminus (s b) (otimes (sub a b) as) bs
  f LT = oiminus (s a) as (itimes (sub b a) bs)  
sub x y |i_ x && o_ y = f (cmp a b) where
  (a,as) = isplit x
  (b,bs) = osplit y
  f EQ = iominus (s a) as bs
  f GT = iominus (s b) (itimes (sub a b) as) bs
  f _ = iominus (s a) as (otimes (sub b a) bs)      
sub x y |i_ x && i_ y = f (cmp a b) where
  (a,as) = isplit x
  (b,bs) = isplit y
  f EQ = iminus (s a) as bs
  f GT = iminus (s b) (itimes (sub a b) as) bs
  f LT = iminus (s a) as (itimes (sub b a) bs)  
\end{code}

\subsection{Defining a total order: comparison}
The comparison operation {\tt cmp}
provides a total order (isomorphic to that on $\N$) on our type $\T$.
It relies on {\tt bitsize} computing the number of applications of $o$ and $i$ 
constructing a term in $\T$.
It is part of our mutually recursive functions, to be defined later.

We first observe that only terms of the same bitsize need detailed comparison,
otherwise the relation between their bitsizes is enough, {\em recursively}.
More precisely, the following holds:
\begin{prop} \label{bitineq}
Let {\tt bitsize} count the number of applications of $o$ and $i$ operations
on a bijective base-2 number. 
Then {\tt bitsize}$(x) <${\tt bitsize}$(y) \Rightarrow x<y$.
\end{prop}
\begin{proof}
Observe that their lexicographic enumeration ensures that the
bitsize of bijective base-2 numbers is a non-decreasing function.
\end{proof}
\begin{code}
cmp E E = EQ
cmp E _ = LT
cmp _ E = GT
cmp x y | x' /= y'  = cmp x' y' where
  x' = bitsize x
  y' = bitsize y
cmp x y = 
  compBigFirst (reversedDual x) (reversedDual y)
\end{code}
The function {\tt compBigFirst} compares two terms known to have the same
{\tt bitsize}. It works on reversed (big digit first) variants,
computed by {\tt reversedDual} and it takes  advantage of the
block structure using the following proposition:
\begin{prop}
Assuming two terms of the same bitsizes, the one starting with
$i$ is larger than one starting with $o$.
\end{prop}
\begin{proof}
Observe that ``big digit first'' numbers are lexicographically ordered with $o<i$.
\end{proof}

As a consequence, {\tt cmp} only recurses when {\em identical} blocks
head the sequence of blocks, otherwise it infers the {\tt LT} or {\tt GT}
relation. 
\begin{code}  
compBigFirst E E = EQ
compBigFirst x y | o_ x && o_ y = f (cmp a b) where
    (a,c) = osplit x
    (b,d) = osplit y
    f EQ = compBigFirst c d
    f LT = GT
    f GT = LT   
compBigFirst x y | i_ x && i_ y = f (cmp a b) where
    (a,c) = isplit x
    (b,d) = isplit y
    f EQ = compBigFirst c d
    f other = other
compBigFirst x y | o_ x && i_ y = LT
compBigFirst x y | i_ x && o_ y = GT
\end{code}
The function reversedDual reverses the order of application
of the $o$ and $i$ operations to a ``biggest digit
first'' order. For this, it only needs to reverse
the order of the alternative blocks of $o^k$ and $i^l$.
It uses the function {\tt len} to compute the number of these
blocks and infer that if odd, the last block is the same
as the first and otherwise it is its alternate.
\begin{code}
reversedDual E = E
reversedDual (V x xs) = f (len (y:ys)) where
  (y:ys) = reverse (x:xs)
  f l | o_ l = V y ys
  f l | i_ l = W y ys
reversedDual (W x xs) = f (len (y:ys)) where
  (y:ys) = reverse (x:xs)
  f l | o_ l = W y ys
  f l | i_ l = V y ys
\end{code}
\begin{code}  
len [] = E
len (_:xs) = s (len xs)  
\end{code}

And based on {\tt cmp}, one can define the minimum {\tt min2},
maximum {\tt max2} 
the absolute value of the difference {\tt absdif}
functions as follows:
\begin{code}
min2 x y = if LT==cmp x y then x else y
max2 x y = if LT==cmp x y then y else x
absdif x y = if LT == cmp x y then sub y x else sub x y
\end{code}

\subsection{Computing {\tt dual} and {\tt bitsize}}
The function {\tt dual} flips {\tt o} and {\tt i} operations
for a natural number seen as written in bijective base 2. 
Note that  with our tree representation it
is constant time, as it simply
flips once the constructors {\tt V} and {\tt W}.
\begin{code}
dual E = E
dual (V x xs) = W x xs
dual (W x xs) = V x xs
\end{code}

The function {\tt bitsize} computes  the number of applications of the 
{\tt o} and {\tt i} operations. It works by summing up (using
Haskell's {\tt foldr}) the counts of {\tt o} and {\tt i} operations
composing a tree-represented natural number.
\begin{code}
bitsize E = E
bitsize (V x xs) = s (foldr add1 x xs)
bitsize (W x xs) = s (foldr add1 x xs)

add1 x y = s (add x y)
\end{code}
Note that {\tt bitsize} also provides an efficient implementation
of the integer $log_2$ operation {\tt ilog2}.
\begin{code}
ilog2 x = bitsize (s' x)
\end{code}

\subsection{Fast multiplication by an exponent of 2}

The function {\tt leftshiftBy} operation uses the fact
that repeated application of the {\tt o} operation ({\tt otimes}) ,  
provides an efficient implementation of 
multiplication with an exponent of 2.
\begin{code}
leftshiftBy _ E = E
leftshiftBy n k = s (otimes n (s' k))
\end{code}

The following holds:
\begin{prop}
Assuming {\tt s,s'} constant time,
{\tt leftshiftBy} is (roughly) logarithmic in the bitsize of 
its arguments.
\end{prop}
\begin{proof}
it follows by observing that at most one addition on data
logarithmic in the bitsize of the operands is performed.
\end{proof}

\subsection{Fast division by an exponent of 2}
Division by an exponent of 2 (equivalent to the {\em rightshift operation}
is more intricate. 
It takes advantage of identities (\ref{onk}) and (\ref{ink})
in a way that is similar to {\tt add} and {\tt sub}.
First, the function {\tt toShift} transforms the outermost block
of $o^m$ or $i^m$ applications to to a multiplication
of the form $k2^m$. It also remembers if it had $o^m$ or $i^m$,
as the first component of the triplet it returns. Note that,
as a result, the operation is actually reversible.
\begin{code}
toShift x | o_ x = (o E,m,k) where
  (a,b) = osplit x
  m = s a
  k = s b
toShift x | i_ x = (i E,m,k) where
  (a,b) = isplit x
  m = s a
  k = s (s b)
\end{code}
Next the function {\tt rightshiftBy} goes over its argument {\tt k}
one block at a time, by comparing the size of the block and its argument
{\tt m} that is decremented after each block by the size of the block.
The local function {\tt f} handles the details, according to
the nature of the block ($o^m$ or $i^m$), and stops when the
argument is exhausted.
More precisely, based on the result {\tt EQ, LT, GT} of the comparison,
as well on the type of block (as recognized by {\tt o\_ p} and {\tt i\_ p}),
it applies back {\tt otimes} or {\tt itimes} when the block
is larger than the value of {\tt m}. Otherwise, it calls itself with
the value of {\tt m} reduced by the size to the block as its first argument.
\begin{code}
rightshiftBy _ E = E
rightshiftBy m k = f (cmp m a) where
  (p,a,b) = toShift k
  
  f EQ | o_ p = sub b p
  f EQ | i_ p = s (sub b p)
  f LT | o_ p = otimes (sub a m) (sub b p)
  f LT | i_ p = s (itimes (sub a m) (sub b p))
  f GT =  rightshiftBy (sub m a) b
\end{code}
The following example illustrates the fact that {\tt rightShiftBy} inverts
the result of{\tt leftshiftBy} on a very large number.
\begin{codex}
*HBin> s' (exp2 (t 100000))
V (V (W E [E]) [E,E,E,E,V E [],V (V E []) [],E]) []
*HBin> leftshiftBy (t 1000) it
W E [W (W E []) [V E [],V (V E []) []],W (W E [E]) 
    [V E [],E,E,V E [],V (V E []) [],E]]
*HBin> rightshiftBy (t 1000) it
V (V (W E [E]) [E,E,E,E,V E [],V (V E []) [],E]) []
\end{codex}

\subsection{Reduced complexity general multiplication}

Devising a similar optimization as for {\tt add} and {\tt sub} for multiplication
is actually easier.
\begin{prop}
The following holds:
\begin{equation} \label{muleq}
o^n(a)o^m(b)=o^{n+m}(ab+a+b)-o^n(a)-o^m(b)
\end{equation}
\end{prop}
\begin{proof}
By (\ref{onk}), we can expand and then reduce:
$o^n(a)o^m(b)=
(2^n(a+1)-1)(2^m(b+1)-1)=
2^{n+m}(a+1)(b+1)-(2^n(a+1)+2^m(b+1)+1=
2^{n+m}(a+1)(b+1)-1-(2^n(a+1)-1+2^m(b+1)-1+2)+2=
o^{n+m}(ab+a+b+ 1)-(o^n(a)+o^m(b))-2+2=
o^{n+m}(ab+a+b)-o^n(a)-o^m(b)$
\end{proof}

The corresponding Haskell code starts with the obvious base cases:
\begin{code}
mul _ E = E
mul E _ = E
\end{code}
When both terms represent odd numbers we apply the identity (\ref{muleq}):
\begin{code}
mul x  y | o_ x && o_ y = r2 where
  (n,a) = osplit x
  (m,b) = osplit y
  p1 = add (mul a b) (add a b)
  p2 = otimes (add (s n) (s m)) p1
  r1 = sub p2 x
  r2 = sub r1 y  
\end{code}
The other cases are reduced to the previous one by using the identity
$i=s.o$.
\begin{code}  
mul x y | o_ x && i_ y = add x (mul x (s' y))
mul x y | i_ x && o_ y = add y (mul (s' x) y)
mul x y | i_ x && i_ y = 
  s (add (add x' y') (mul x' y')) where 
    x'=s' x 
    y'=s' y
\end{code}
Note that when the operands are composed of large blocks of alternating
$o^n$ and $i^m$ applications, the algorithm is quite efficient as it
works (roughly) in time proportional to the number of blocks rather
than the number of digits.  
The following example illustrates a blend of 
arithmetic operations benefiting from complexity reductions
on giant tree-represented numbers:
\begin{codex}
*HBin> let term1 = sub (exp2 (exp2 
                    (t 12345))) (exp2 (t 6789))
*HBin> let term2 = add (exp2 (exp2 (t 123))) 
                    (exp2 (t 456789))
*HBin> ilog2 (ilog2 (mul term1 term2))
V E [E,E,W E [],V E [E],E]
*HBin> n it
12345
\end{codex}
This opens {\em a new world where arithmetic operations are
not limited by the size of their operands, but only by
their ``structural complexity''}. We will make this concept
more precise in section \ref{stru}.

\subsection{Power}

We first specialize our multiplication for a slightly
faster squaring operation, using the identity:

\begin{equation} \label{sq}
(o^n(a))^2=o^{2n}(a^2+2a)-2o^n(a)
\end{equation}

\begin{code}
square E = E
square x | o_ x = r where
  (n,a) = osplit x
  p1 = add (square a) (db a)
  p2 = otimes (i n) p1
  r = sub p2 (db x) 
square x| i_ x = s (add (db x') (square x')) where 
  x' = s' x

\end{code}

We can implement a simple but fairly efficient
`` power by squaring'' operation for $x^y$ as follows:

\begin{code}
pow _ E = V E []
pow x y | o_ y = mul x (pow (square x) (o' y))
pow x y | i_ y = mul x2 (pow x2 (i' y)) where 
  x2 = square x
\end{code}

It works well with fairly large numbers, by also
benefiting from efficiency of multiplication on terms
with large blocks of  $o^n$ and $o^m$ applications:
\begin{codex}
*HBin> n (bitsize (pow (t 2014) (t 100)))
1097
*HBin> pow (t 32) (t 10000000)
W E [W (W (V E []) []) [W E [E],
  V (V E []) [],E,E,E,W E [E],E]]
\end{codex}

\section{Structural complexity} \label{stru}

As a measure of structural complexity we define
the function {\tt tsize} that counts the nodes
of a tree of type $\T$ (except the root).
\begin{code}
tsize E = E
tsize (V x xs) = foldr add1 E (map tsize (x:xs))
tsize (W x xs) = foldr add1 E (map tsize (x:xs))
\end{code}
It corresponds to the function  $c:\T \to \N$
 defined as follows:
\begin{equation}
c(t)=
\begin{cases}
0&  \text{if $\text{\tt t}=~$0},\\
\sum_{y \in (\text{\tt x:xs})} {(1+ts(y))} &  \text{if t = \text{\tt V x xs}},\\
\sum_{y \in (\text{\tt x:xs})} {(1+ts(y))} &  \text{if t = \text{\tt W x xs}} .
\end{cases}
\end{equation}
The following holds:
\begin{prop} \label{bitcmp}
For all terms $t \in \T$, {\tt tsize t} $\leq$ {\tt bitsize t}.
\end{prop}
\begin{proof}
By induction on the structure of $t$, by observing that the two
functions have similar definitions and corresponding calls to
{\tt tsize} return terms assumed smaller than those of
{\tt bitsize}.
\end{proof}

The following example illustrates their use:
\begin{codex}
*HBin> map (n.tsize.t) [0,100,1000,10000]
[0,6,8,10]
*HBin> map (n.bitsize.t) [0,100,1000,10000]
[0,6,9,13]

*HBin> map (n.tsize.t) [2^16,2^32,2^64,2^256]
[4,5,5,5]
*HBin> map (n.bitsize.t) [2^16,2^32,2^64,2^256]
[16,32,64,256]
\end{codex}

\begin{codeh}
repcomp m =  [bsizes,tsizes] where
  ts = map t [0..2^m-1]
  bsizes = map (n.bitsize) ts
  tsizes = map (n.tsize) ts
  
repdif x = n (sub (bitsize (t x)) (tsize (t x)))  
\end{codeh}

Figure \ref{tsizes} shows the reductions in structural complexity compared with bitsize for an initial interval of $\N$.

\FIG{tsizes}{Structural complexity (yellow line) bounded by bitsize (red line) from $0$ to $2^{10}-1$}{0.20}{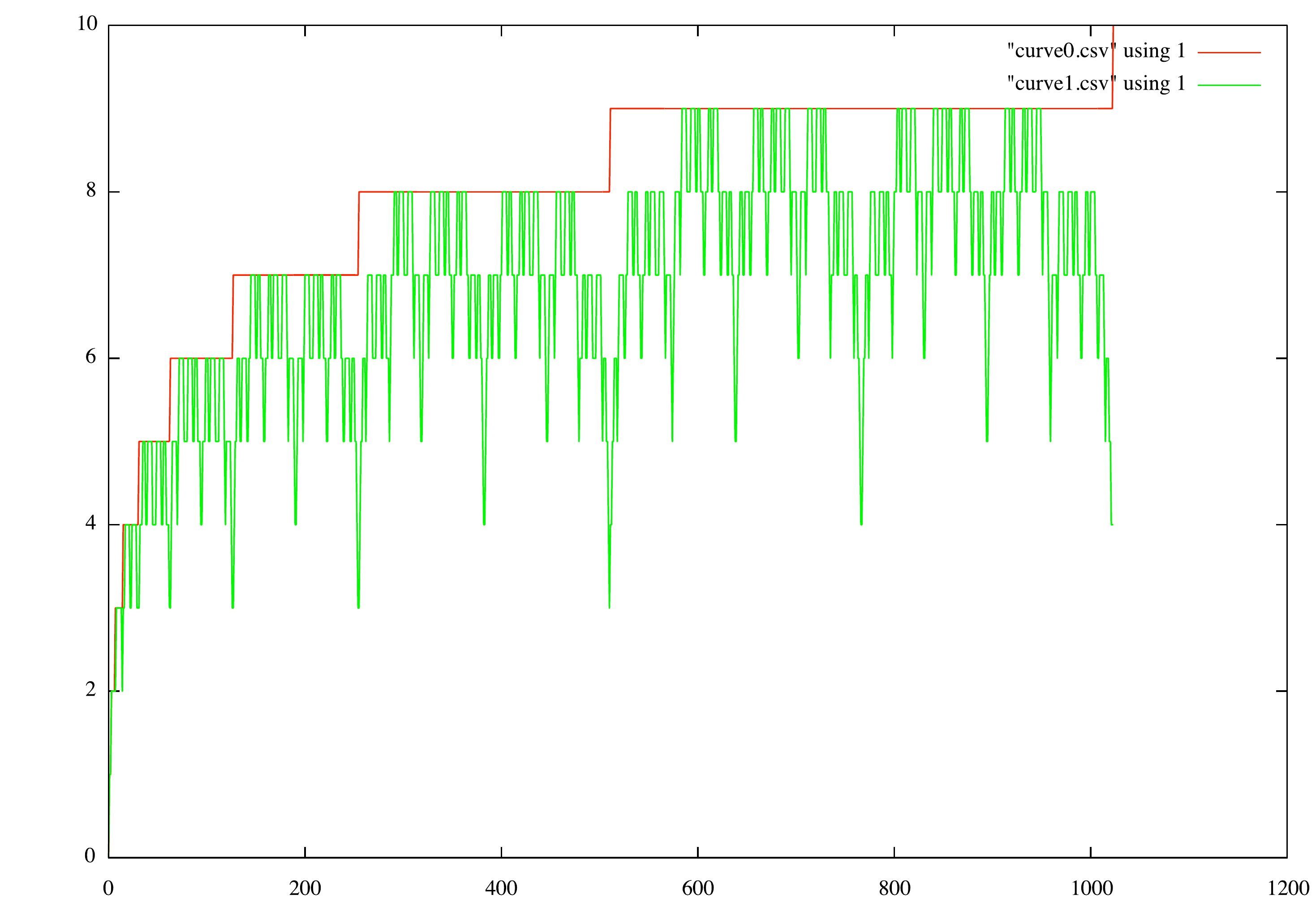}

After defining the function {\tt iterated} that applies f k times
\begin{code}  
iterated f E x = x
iterated f k x = f (iterated f (s' k) x) 
\end{code}
we can exhibit a best case
\begin{code}
bestCase k = s' (iterated exp2 k E)
\end{code}
and a worse case
\begin{code}
worseCase k = iterated (i.o) k E 
\end{code}
The following examples illustrate these functions:
\begin{codex}
*HBin> bestCase (t 5)
V (V (V (V E []) []) []) []
*HBin> n it
65535
*HBin> bestCase (t 5)
V (V (V (V E []) []) []) []
*HBin> n (bitsize (bestCase (t 5)))
16
*HBin> n (tsize (bestCase (t 5)))
4

*HBin> worseCase (t 5)
W E [E,E,E,E,E,E,E,E,E]
*HBin> n it
1364
*HBin> n (bitsize (worseCase (t 5)))
10
*HBin> n (tsize (worseCase (t 5)))
10
\end{codex}
The function {\tt bestCase } computes the iterated
exponent of 2 (tetration) and then applies the predecessor {\tt s'}
to it. A simple closed formula can also be found for {\tt worseCase}: 
\begin{prop}
The function {\tt worseCase k} computes the value in $\T$
corresponding to the value ${{4(4^{k} - 1)} \over 3} \in \N$.
\end{prop}
\begin{proof}
By induction or by 
applying the iterated function formula to $f(x)=i(o(x)))=2(2x+1)+2=4(x+1)$.
\end{proof}

The average space-complexity of the representation
is related to the average length of the {\em integer partitions of
the bitsize of a number} \cite{part99}. Intuitively, the shorter the
partition in alternative blocks of $o$ and $i$ applications,
the more significant the compression is, but the exact study,
given the recursive application of run-length encoding,
is likely to be quite intricate.

Note also 
that our concept of structural complexity is only
a weak approximation
of Kolmogorov complexity \cite{vitanyi}.
For instance, the reader might
notice that  our worse case example
is computable by a program of relatively
small size. However, as {\tt bitsize}
is an upper limit to {\tt tsize}, we can
be sure that we are within constant
factors from the corresponding bitstring
computations, even on random data
of high Kolmogorov complexity.

Note also that an alternative concept of structural
complexity can be defined by considering the (vertices+edges) size of the DAG obtained
by folding together identical subtrees. We will use such DAGs in section \ref{theo}
to more compactly visualize large tree-represented numbers.

As section \ref{theo} will illustrate it, several interesting
number theoretical entities that hold current records in
various categories have very low structural complexities,
contrasting to gigantic bitsizes.

\section{Efficient representation of some important number-theoretical entities} \label{theo}


Let's first observe that Fermat, Mersenne and perfect numbers have all compact
expressions with our tree representation of type $\T$.

\begin{code}
fermat n = s (exp2 (exp2 n))

mersenne p = s' (exp2 p)

perfect p = s (V q [q]) where q = s' (s' p)
\end{code}

\begin{codex}
*HBin> mersenne 127
170141183460469231731687303715884105727
*HBin> mersenne (t 127)
V (W (V E [E]) []) []
\end{codex}
The largest known prime number, found by the GIMPS distributed computing 
project \cite{gimps} in January 2013 is the 
48-th Mersenne prime = $2^{57885161}-1$ (with possibly
smaller Mersenne primes below it).
It is defined in Haskell as follows:
\begin{code}
-- exponent of the 48-th known Mersenne prime
prime48 = 57885161
-- the actual Mersenne prime
mersenne48 = s' (exp2 (t prime48))
\end{code}
While it has a bit-size of 57885161, its 
compressed tree representation is rather small:
\begin{codex}
*HBin> mersenne48
V (W E [V E [],E,E,V (V E []) [],
   W E [E],E,E,V E [],V E [],W E [],E,E]) []
\end{codex}
The equivalent DAG representation
of the 48-th Mersenne prime,
shown in Figure \ref{mersenne48}, has only 7 shared nodes 
and structural complexity 22. Note that the empty leaf node 
is marked with the letter {\tt T}.
\FIG{mersenne48}{Largest known prime number discovered in January 2013: the 48-th Mersenne prime, represented as a DAG}{0.30}{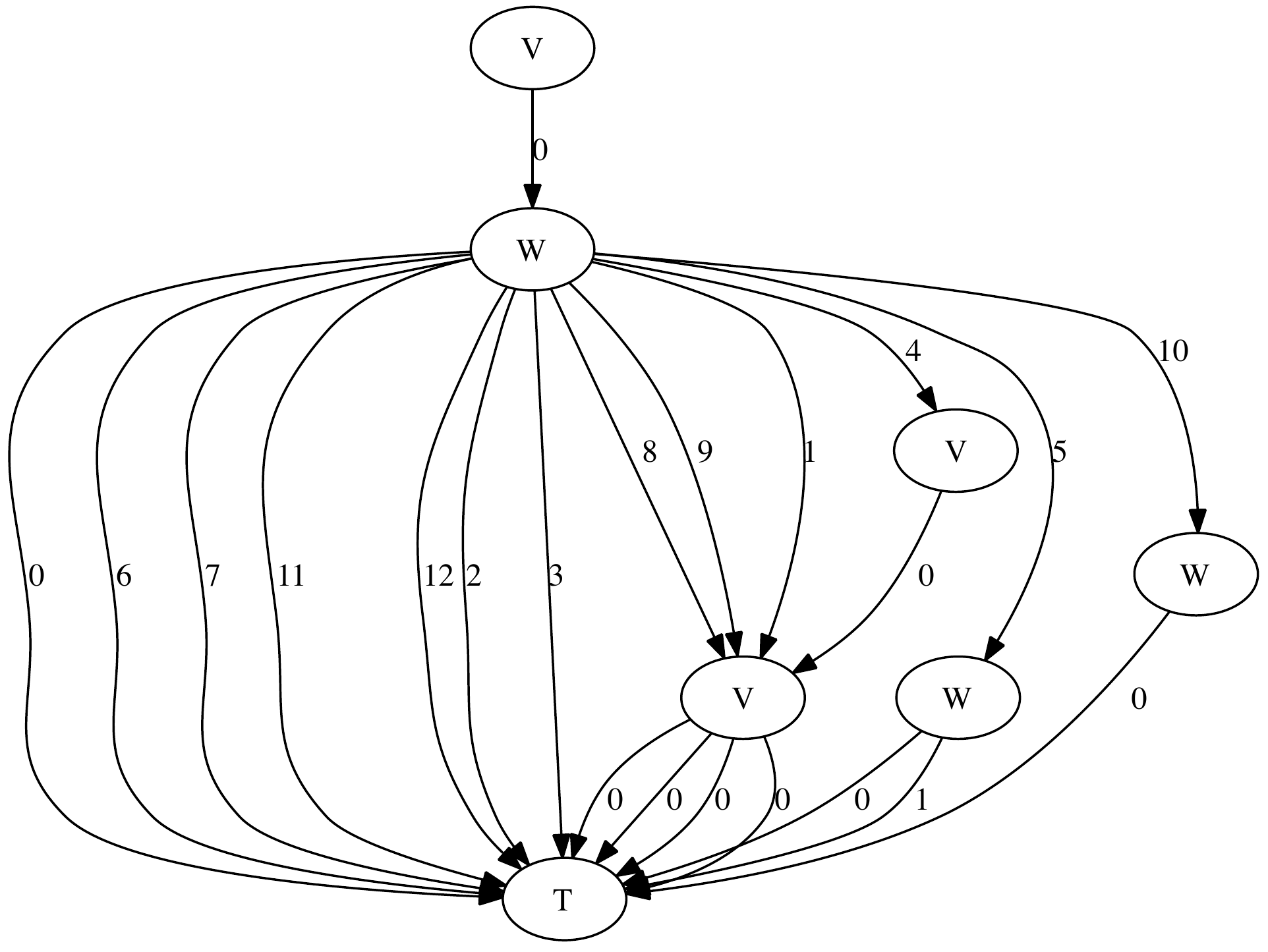}

It is interesting to note that similar compact representations can also
be derived for perfect numbers. For instance, the largest known perfect number, derived
from the largest known Mersenne prime as $2^{57885161-1}(2^{57885161}-1)$,
(involving only 8 shared nodes and structural complexity 43) is:
\begin{code}
perfect48 = perfect (t prime48)
\end{code}
Fig. \ref{perfect48} shows the DAG representation of the largest known perfect number,
derived from Mersenne number 48.
\FIG{perfect48}{Largest known perfect number in January 2013}{0.25}{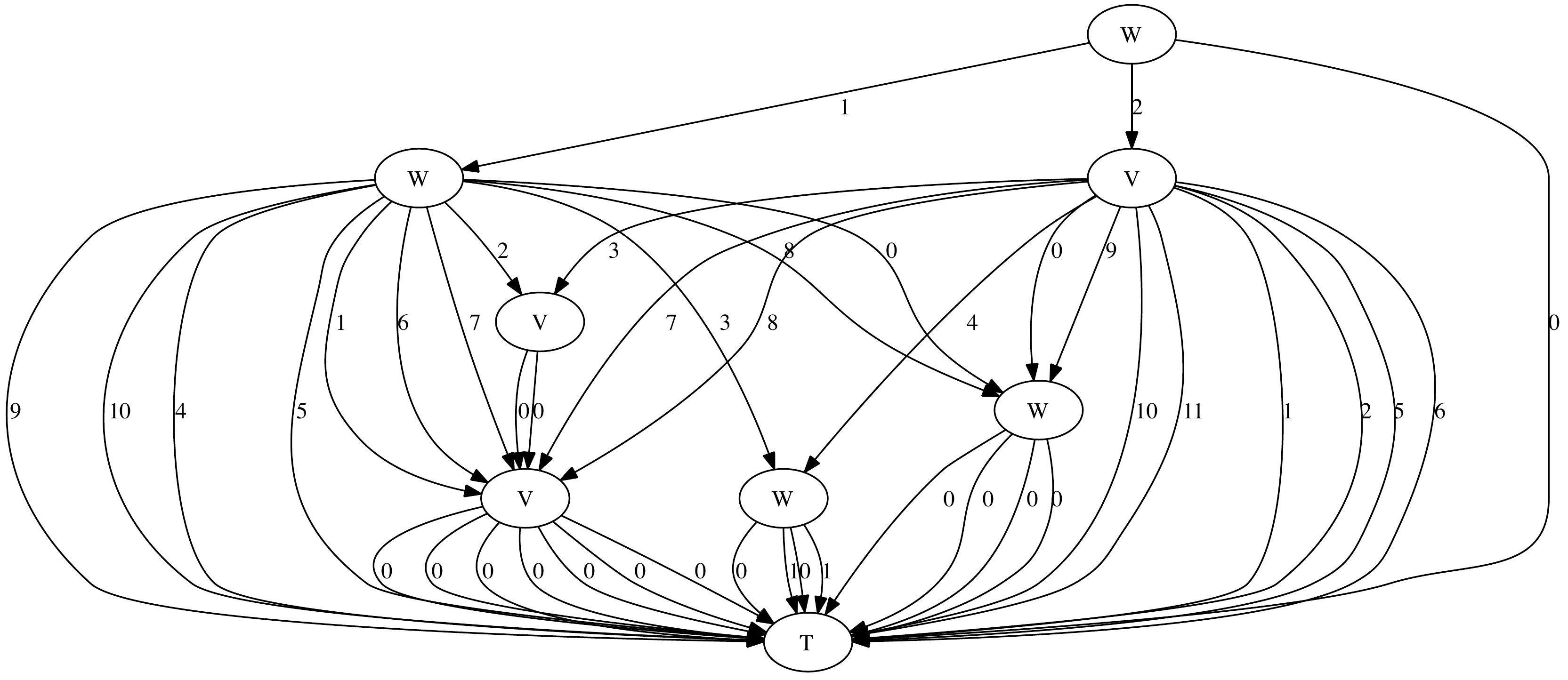}
Similarly, the largest Fermat number that has been factored so far, $F_{11}=2^{2^{11}}+1$ is compactly represented as 
\begin{codex}
*HBin> fermat (t 11)
V E [E,V E [W E [V E []]]]
\end{codex}
with structural complexity 8.
By contrast, its (bijective base-2) binary representation consists of 2048 digits.

Some other very large primes that are not 
Mersenne numbers also have compact representations.

The generalized Fermat prime $27653*2^{9167433}+1$, (currently the 15-the largest prime number)  
computed as a tree is:
\begin{code}
genFermatPrime = s (leftshiftBy n k) where 
  n = t (9167433::Integer)
  k = t (27653::Integer)
\end{code}
\begin{codex}
*HBin> genFermatPrime
V E [E,W (W E []) [W E [],E,
  V E [],E,W E [],W E [E],E,E,W E []],
  E,E,E,W (V E []) [],V E [],E,E]     
\end{codex}
Figure \ref{genFermatPrime} shows the DAG representation of this generalized Fermat prime
with 7 shared nodes and structural complexity 30.
\FIG{genFermatPrime}
{Generalized Fermat prime}
{0.30}{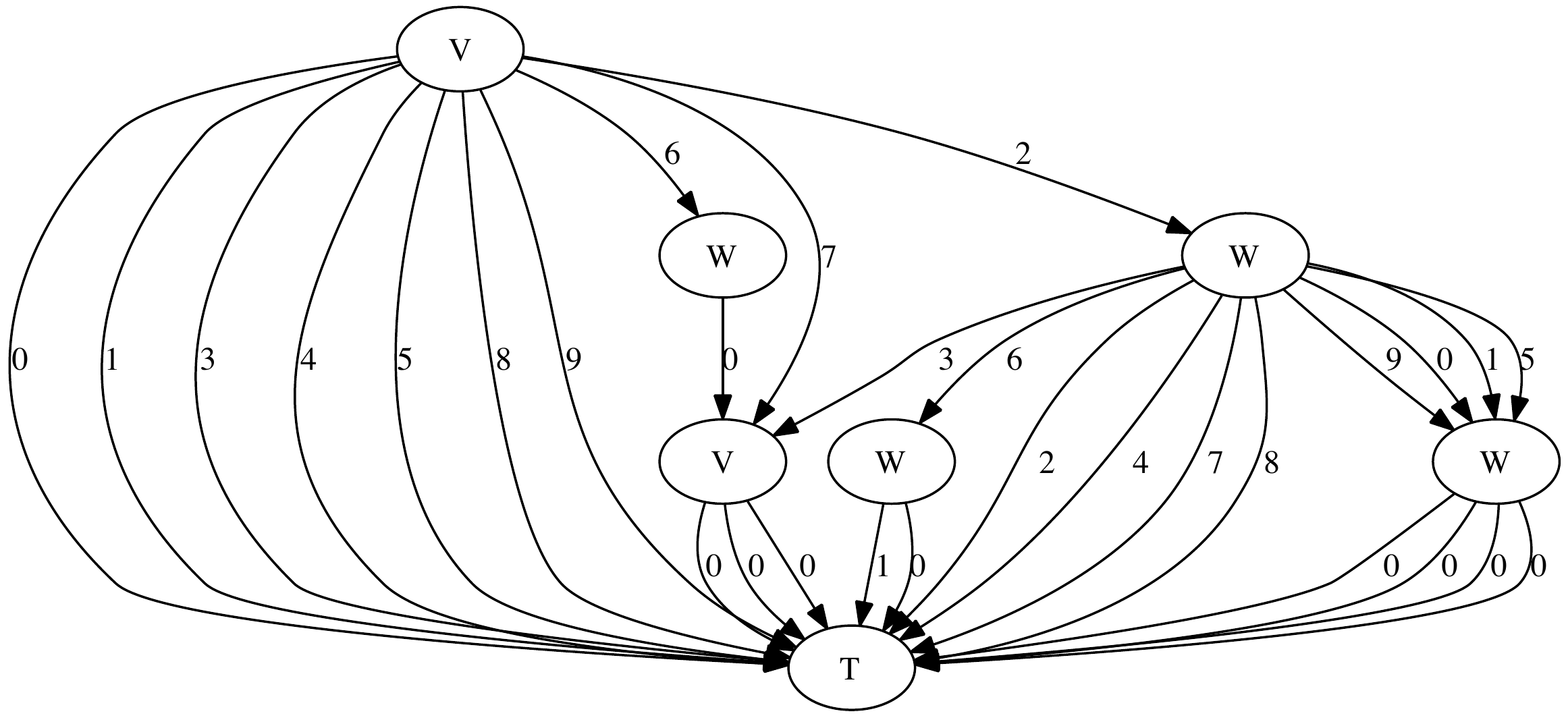}

The largest known Cullen prime $6679881*2^{6679881}+1$ computed as tree
(with 6 shared nodes and structural complexity 43) is:
\begin{code}
cullenPrime = s (leftshiftBy n n) where 
  n = t (6679881::Integer)
\end{code}
\begin{codex}
*HBin> cullenPrime
V E [E,W (W E []) [W E [],E,E,E,E,V E [],E,
  V (V E []) [],E,E,V E [],E],E,V E [],E,V E [],
  E,E,E,E,V E [],E,V (V E []) [],E,E,V E [],E]
\end{codex}
Figure \ref{cullenPrime} shows the DAG representation of this Cullen prime.
\FIG{cullenPrime}
{largest known Cullen prime}
{0.20}{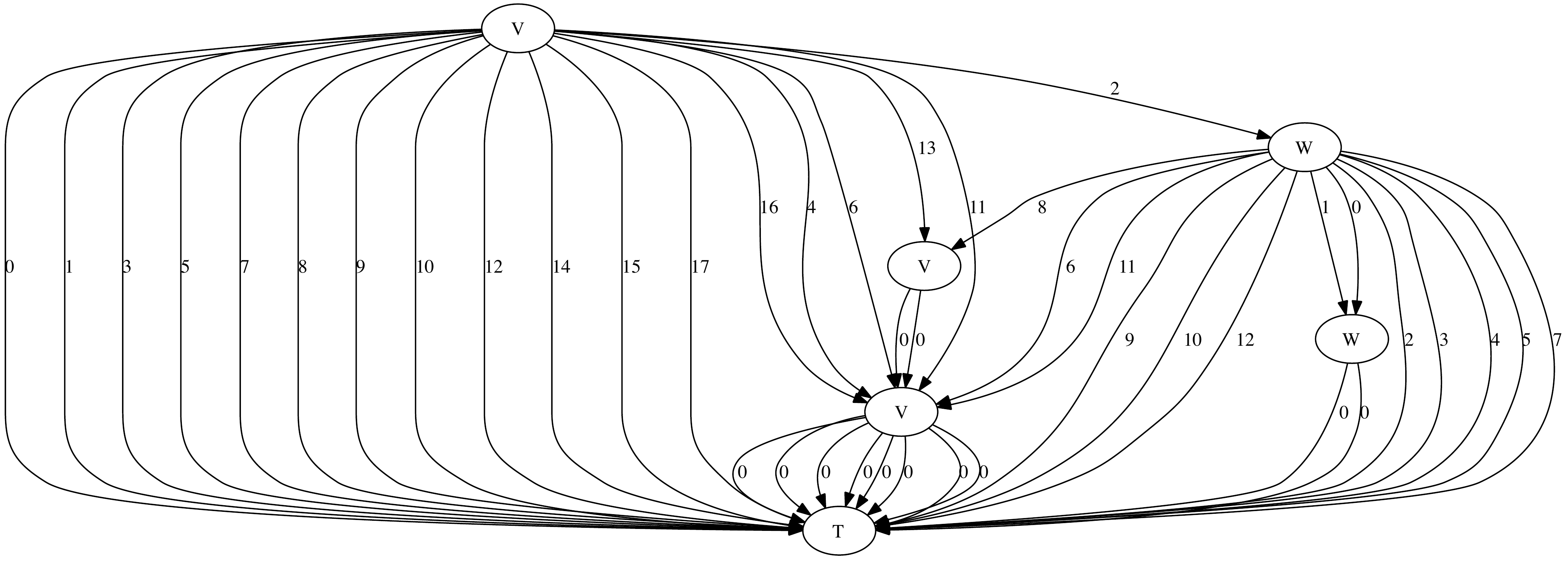}
The largest known Woodall prime $3752948 * 2^{3752948} - 1$
computed as a tree (with 6 shared nodes and structural complexity 33) is:
\begin{code}
woodallPrime = s' (leftshiftBy n n) where 
  n = t (3752948::Integer)
\end{code}
\begin{codex}
*HBin> woodallPrime
V (V E [V E [],E,V E [E],V (V E []) [],
 E,E,E,V E [],V E []]) [E,E,V E [E],
 V (V E []) [],E,E,E,V E [],V E []]
\end{codex}
Figure \ref{woodallPrime} shows the DAG representation of this Woodall prime.
\FIG{woodallPrime}
{Largest known Woodall prime}
{0.30}{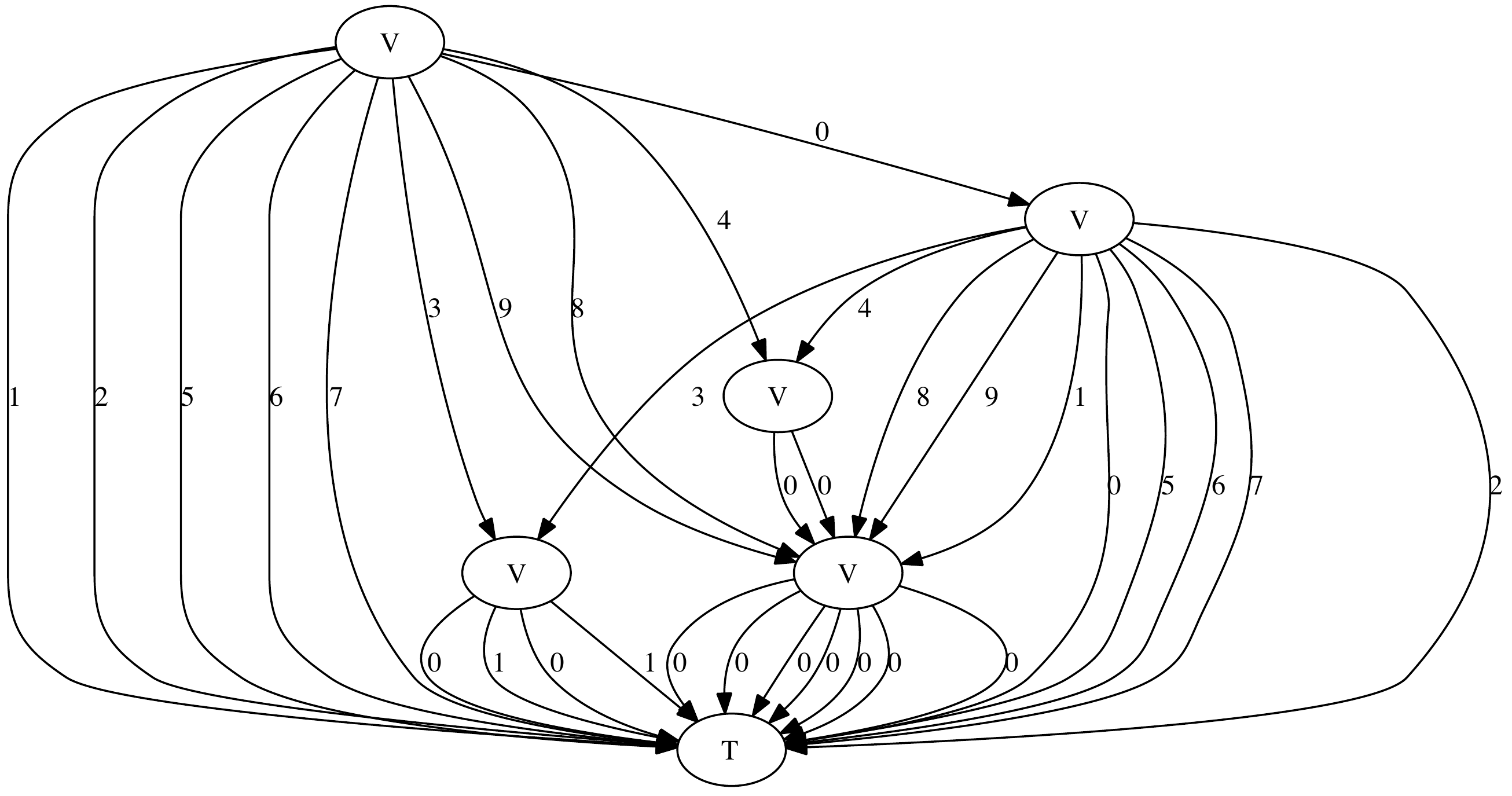}

The largest known Proth prime $19249*2^{13018586} + 1$
computed as a tree is:
\begin{code}
prothPrime = s (leftshiftBy n k) where 
  n = t (13018586::Integer)
  k = t (19249::Integer)
\end{code}
\begin{codex}
*HBin> prothPrime
V E [E,V (W E []) [V E [],E,W E [],E,E,
  V E [],E,E,E,E,V E [],W E [],E],E,W E [],
  V E [],V E [],V E [],E,E,V E []]
\end{codex}
Figure \ref{prothPrime} shows the DAG representation of this Proth prime,
the largest non-Mersenne prime known by March 2013 with 5 shared nodes and structural
complexity 36.
\FIG{prothPrime}
{Largest known Proth prime}
{0.25}{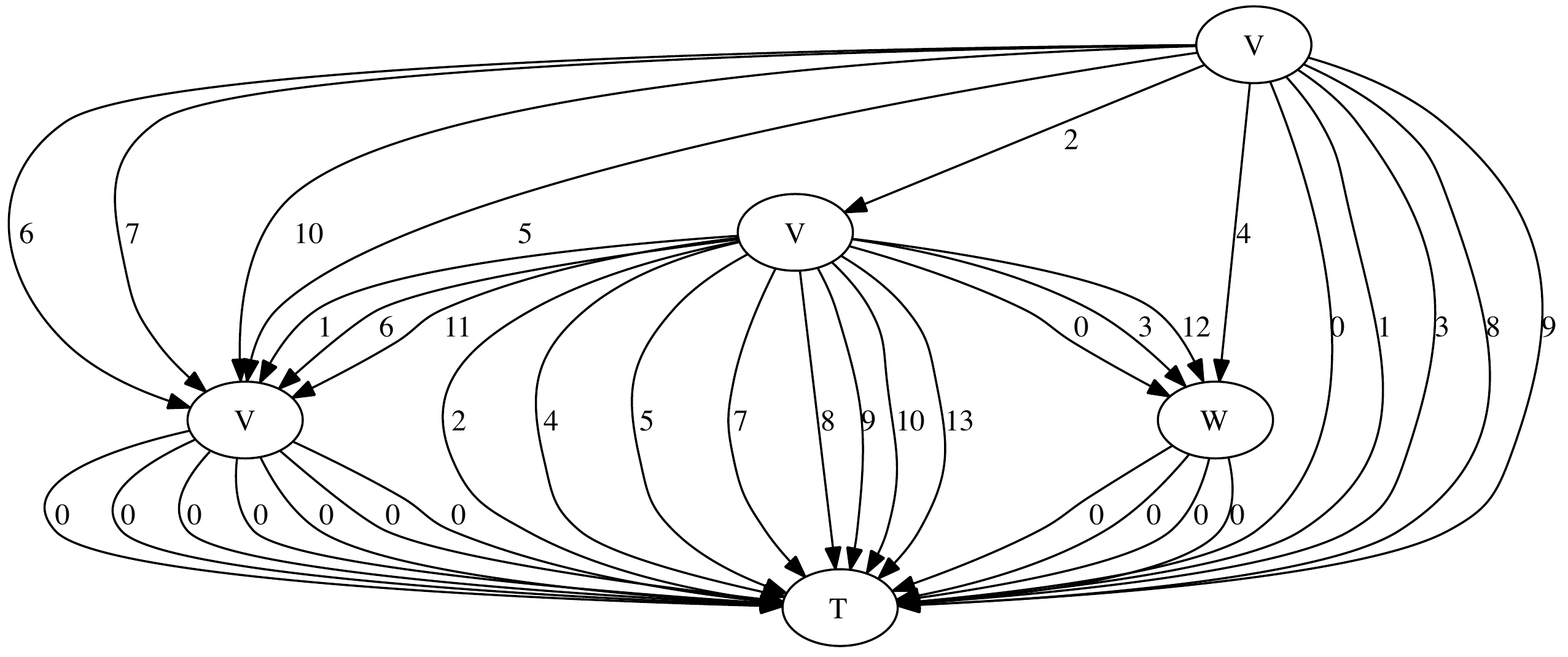}

The largest known Sophie Germain prime $18543637900515*2^{666667}-1$
computed as a tree (with 6 shared nodes and structural complexity 56) is:
\begin{code}
sophieGermainPrime = s' (leftshiftBy n k) where 
  n = t (666667::Integer)
  k = t (18543637900515::Integer)
\end{code}
\begin{codex}
*HBin> sophieGermainPrime
V (W (V E []) [E,E,E,E,V (V E []) [],V E [],E,
   E,W E [],E,E]) [V E [],W E [],W E [],V E [],
   V E [],E,E,V E [],V E [],V E [],V (V E []) 
   [],E,V E [],V (V E []) [],V E [],E,W E [],E,
   V E [],V (V E []) []]
\end{codex}
Figure \ref{sophieGermainePrime} shows the DAG representation of this prime.
\FIG{sophieGermainePrime}
{Largest known Sophie Germain prime}
{0.20}{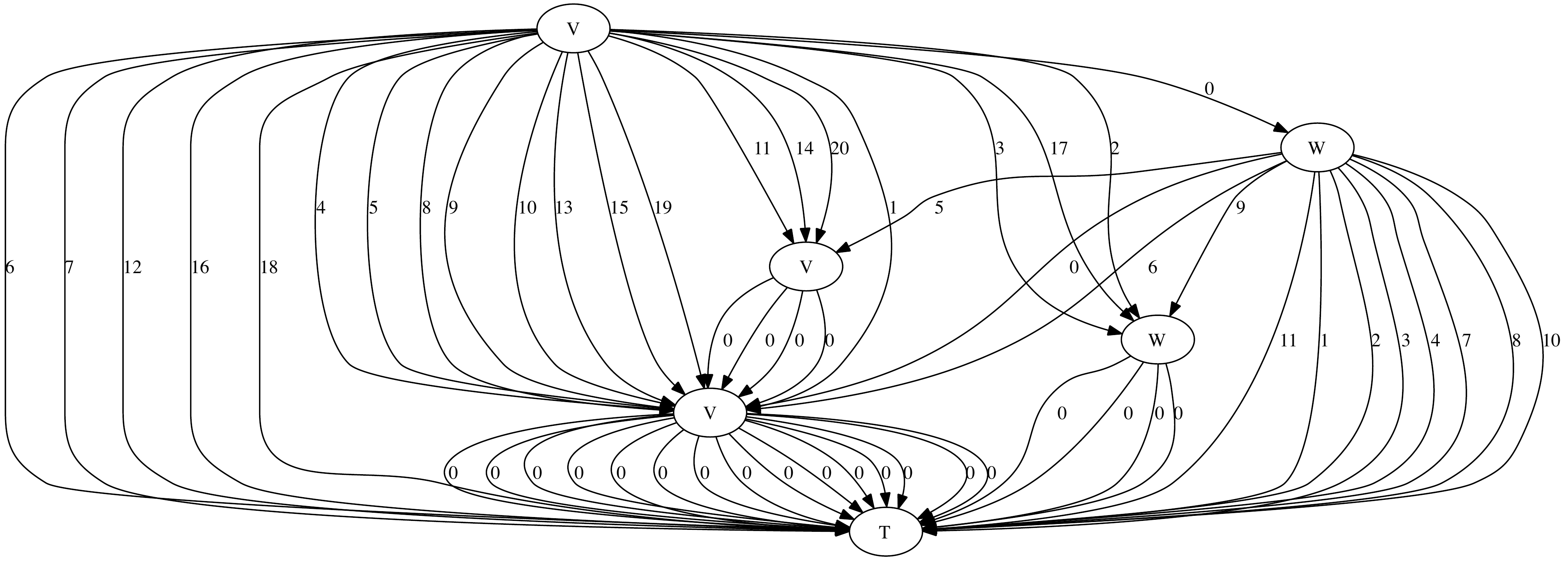}
The largest known twin primes $3756801695685*2^{666669} \pm 1$ computed as a pair of
trees (with 7 shared nodes both and structural complexities of 54 and 56) are:
\begin{code}
twinPrimes = (s' m,s m) where 
  n = t (666669::Integer)
  k = t (3756801695685::Integer)
  m = leftshiftBy n k
\end{code}
\begin{codex}
*HBin> fst twinPrimes
V (W E [E,V E [],E,E,V (V E []) [],
   V E [],E,E,W E [],E,E]) 
   [E,E,E,W E [],W (V E []) [],V E [],E,V E [],
   E,E,E,E,V E [],E,E,
   V E [],V E [],E,E,E,E,E,E,E,V E [],E,E]
*HBin> snd twinPrimes
V E [E,W (V E []) [E,E,E,E,V (V E []) [],
  V E [],E,E,W E [],E,E],E,E,E,
  W E [],W (V E []) [],V E [],E,V E [],
  E,E,E,E,V E [],E,E,V E [],
  V E [],E,E,E,E,E,E,E,V E [],E,E]
\end{codex}
Figures \ref{twinPrimes1} and \ref{twinPrimes2} show 
the DAG representation of these twin primes.
\FIG{twinPrimes1}
{Largest known twin prime 1}
{0.13}{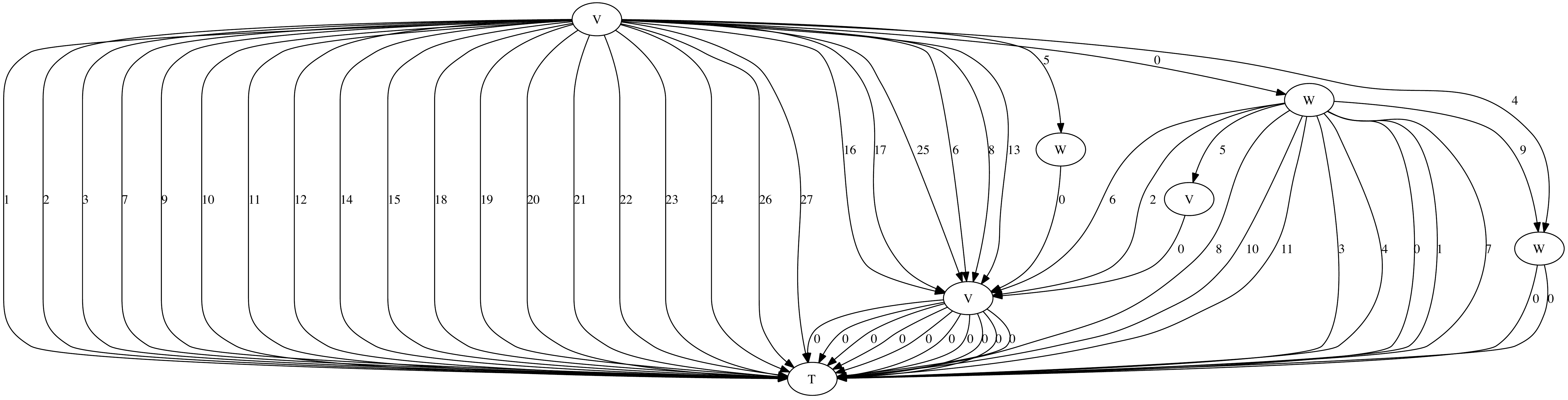}
\FIG{twinPrimes2}
{Largest known twin prime 2}
{0.13}{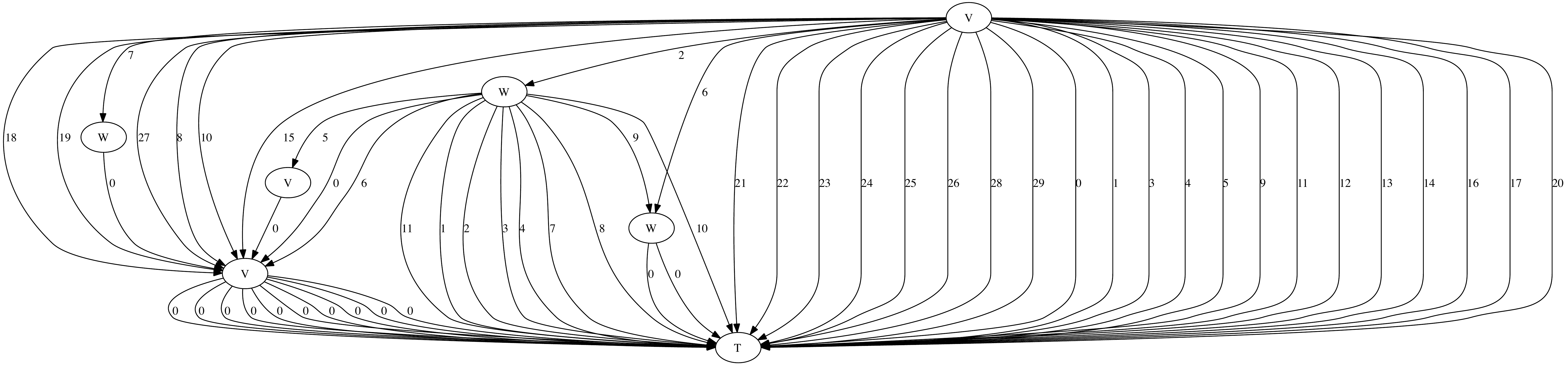}
One can appreciate the succinctness of our representations,
given that all these numbers have hundreds of thousands 
or millions of decimal digits. An interesting challenge would
be to (re)focus on discovering primes with significantly larger 
structural complexity then the current record holders by bitsize.

More importantly, as the following examples illustrate it,
{\em computations} like addition, subtraction and multiplication of
such numbers are possible:
\begin{codex}
*HBin> sub genFermatPrime (t 2014)
V (V E []) [E,V E [],E,W E [E],V E [W E [E],
  W E [],E,W E [],W E [E],E,E,W E []],V E [],
  E,W (V E []) [],V E [],E,E]
*HBin> bitsize (sub prothPrime (t 1234567890))
W E [V E [],E,E,V (V E []) [],E,E,
  V E [],E,E,E,E,V E [],W E [],E]
*HBin> tsize (exp2 (exp2 mersenne48))
V E [E,E,E]
*HBin> tsize(leftshiftBy mersenne48 mersenne48)
V E [W E [],E]
*HBin> add (t 2) (fst twinPrimes) == 
           (snd twinPrimes)
True
*HBin> ilog2 (ilog2 (mul prothPrime cullenPrime))
W E [V E [],E]
*HBin> n it
24
\end{codex}

\section{Computing the Collatz/Syracuse sequence for huge numbers}
As an interesting application, that achieves something one cannot
do with ordinary Integers is to explore the behavior of interesting
conjectures in the new world of numbers limited not by their
sizes but by their structural complexity.
The Collatz conjecture states that the function
\begin{equation}
collatz(x)=
\begin{cases}
x \over 2  &  \text{if $x$ is even},\\
3x+1 &  \text{if $x$ is odd}.\\
\end{cases}
\end{equation}
reaches $1$ after a finite number of iterations.
An equivalent formulation, by grouping together all
the division by 2 steps, is the function:
\begin{equation}
collatz'(x)=
\begin{cases}
x \over 2^{\nu_2(x)}  &  \text{if $x$ is even},\\
3x+1 &  \text{if $x$ is odd}.\\
\end{cases}
\end{equation}
where $\nu_2(x)$ denotes the {\em dyadic valuation of x}, i.e., the largest
exponent of 2 that divides x. One step further, the
{\em syracuse function} is defined as the odd integer $k'$ such
that $n=3k+1=2^{\nu(n)}k'$. One more step further, by writing $k'=2m+1$
we get a function that associates $k \in \N$ to $m \in \N$.

The function {\tt tl} computes efficiently the equivalent
of 
\begin{equation}
tl(k) = {{{k \over 2^{\nu_2(k)}}-1} \over 2}
\end{equation}
\begin{code}      
tl n | o_ n = o' n 
tl n | i_ n = f xs where 
  V _ xs = s' n
  f [] = E
  f (y:ys) = s (i' (W y ys)) 
\end{code} 
\begin{codeh}
cons n y = s (otimes n (s' (o y)))

hd z | o_ z = E
hd z = s x where V x _ = s' z
\end{codeh}
Then our variant of the {\em syracuse function} corresponds to
\begin{equation}
syracuse(n) = tl(3n+2)
\end{equation}
which we can code efficiently as
\begin{code}
syracuse n = tl (add n (i n))
\end{code}
The function {\tt nsyr} computes the iterates of this function, 
until (possibly) stopping:
\begin{code}
nsyr E = [E]
nsyr n = n : nsyr (syracuse n)
\end{code}

It is easy to see that the Collatz conjecture is true if and only if {\tt nsyr}
terminates for all $n$, as illustrated by the following example:
\begin{codex}
*HBin> map n (nsyr (t 2014))
[2014,755,1133,1700,1275,1913,2870,1076,807,
   1211,1817,2726,1022,383,575,863,1295,1943,
   2915,4373,6560,4920,3690,86,32,24,18,3,5,
   8,6,2,0]
\end{codex}
The next examples will show that computations for
{\tt nsyr} can be efficiently carried out
for numbers that with traditional bitstring notations would easily
overflow even the memory of a computer using as transistors
all the atoms in the known universe.

The following examples illustrate this:
\begin{codex}
map (n.tsize) (take 1000 (nsyr mersenne48))
[22,22,24,26,27,28,...,1292,1313,1335,1353]
\end{codex}
As one can see, the structural complexity is growing
progressively, but that our tree-numbers have no
trouble with the computations.
\begin{codex}
*HBin> map (n.tsize) 
        (take 1000 (nsyr mersenne48))
[22,22,24,26,27,28,...,1292,1313,1335,1353]
\end{codex}
Moreover, we can keep going up with a tower of 3 exponents.
Interestingly, it results in a fairly small increase in structural 
complexity over the first 1000 terms.
\begin{codex}
*HBin> map (n.tsize) (take 1000 
        (nsyr (exp2 (exp2 (exp2 mersenne48)))))
[26,33,36,37,40,42,...,1313,1335,1358,1375]
\end{codex}

While we did not tried to wait out the 
termination of the execution for Mersenne number 48 we,
have computed {\tt nsyr} for the record holder from 1952,
which is still much larger than the values (up to $5 \times  2^{60}$) for which
the conjecture has been confirmed true.
Figure \ref{m15ts} shows the structural complexity curve
for the ``hailstone sequence'' associated
by the function {\tt nsyr}
to the 15-th Mersenne prime, $2^{1279}-1$
\FIG{m15ts}
{Structural complexity curve for {\tt nsyr} on $2^{1279}-1$ }
{0.25}{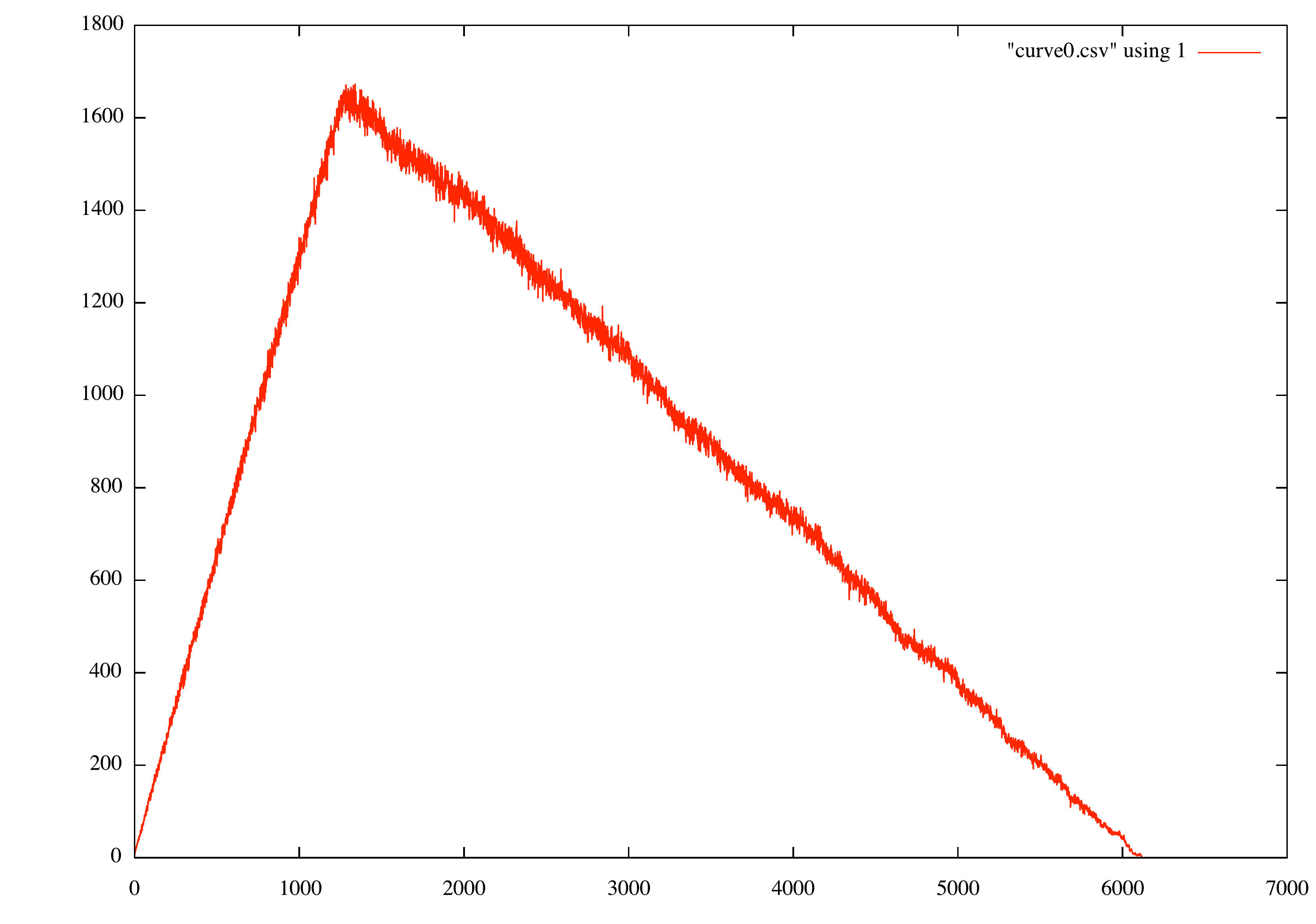}

As an interesting fact,
possibly unknown so far, one
might notice the abrupt phase transition that, based on our
experiments, seem to characterize  the behavior of this function,
when starting with very large numbers of relatively small structural
complexity.

And finally something we are quite sure has never been computed before,
we can also start with a {\em tower of exponents 100 levels tall}:
\begin{codex}
*HBin> take 1000 (map(n.tsize)(nsyr (bestCase (t 100)))) 
[99,99,197,293,294,296,299,299,...,1569,1591,1614,1632]
\end{codex}

\section{Related work} \label{related}

We will briefly describe here some related work
that has inspired and facilitated this line of research
and will help to put our past contributions and planned
developments in context.

Several notations for very large numbers have been invented in the past. Examples
include Knuth's {\em arrow-up} notation \cite{knuthUp}
covering operations like the {\em tetration} (a notation for towers of exponents).
In contrast to our tree-based natural numbers,
such notations are not closed under
addition and multiplication, and consequently
they cannot be used as a replacement
for ordinary binary or decimal numbers.

The first instance of a hereditary number system, at our best knowledge,
occurs in the proof of Goodstein's theorem \cite{goodstein}, where
replacement of finite numbers on a tree's branches by the ordinal $\omega$
allows him to prove that a ``hailstone sequence'' visiting arbitrarily
large numbers eventually turns around and terminates.

Like our trees of type $\T$,
Conway's surreal numbers \cite{surreal} can also be seen as inductively
constructed trees. While our focus is on 
efficient large natural number
arithmetic and sparse set representations,
surreal numbers  model games,
transfinite ordinals and generalizations 
of real numbers.

Numeration systems on regular languages have been studied
recently, e.g. in \cite{Rigo2001469} and specific instances
of them are also known as 
bijective base-k numbers.  
Arithmetic packages similar to 
our bijective base-2 view of arithmetic operations
are part of libraries of proof assistants
like Coq \cite{Coq:manual}.

Arithmetic computations based
on recursive data types like
the free magma of binary trees
(isomorphic to the 
context-free language of balanced parentheses)
are described in \cite{sac12},
where they are seen as G\"odel's {\tt System T} types,
as well as combinator application trees.
In \cite{ppdp10tarau}
a type class mechanism is used
to express computations on hereditarily
finite sets and hereditarily finite
functions.
In \cite{vu09} integer decision diagrams are introduced
providing a compressed representation for sparse
integers, sets and various other data types. 

\section{Conclusion and future work} \label{concl}

We have provided in the form of a literate Haskell program a
declarative specification of a tree-based number system.
Our emphasis here was on the correctness and the
theoretical complexity bounds of our operations
rather than the packaging in a form that would compete
with a C-based arbitrary size integer package like GMP.
We have also ensured that our algorithms are
as simple as possible and
we have closely correlated our Haskell code
with the formulas describing 
the corresponding arithmetical properties.
As the algorithms involved are all novel and we have
explored genuinely uncharted territory, we are not considering
this literate program a {\em functional pearl}, as we are
by no means focusing on polishing known results, but rather
on using the niceties of functional programming
to model new concepts. For instance,
our algorithms rely on properties of blocks of iterated applications
of functions rather than the ``digits as coefficients of polynomials''
view of traditional numbering systems. While the rules are often
more complex, restricting our code
to a purely declarative subset of functional programming
made managing a fairly intricate network of mutually 
recursive dependencies much easier.

We have shown that some interesting number-theoretical entities like
Fermat and perfect numbers, and the largest known
Mersenne, Proth, Cullen, Sophie Germain and twin primes
have compact
representations with our tree-based numbers. One may observe their
common feature: they are all represented in terms of exponents of 2,
successor/predecessor and specialized multiplication operations.

But more importantly, we have shown that {\em computations} like
addition, subtraction, multiplication, bitsize, exponent of 2, that 
favor
giant numbers with {\em low structural complexity}, are performed
in constant time, or time proportional to their
structural complexity. We have also studied the best and worse case
structural complexity of our representations and shown that,
as structural complexity is bounded by bitsize, computations and
data representations are within
constant factors of conventional arithmetic even
in the worse case.

The fundamental theoretical challenge raised at this point is the following:
{\em can other number-theoretically interesting operations
expressed succinctly in terms of our tree-based data type? Is it possible to reduce the
complexity of some other important operations, besides those found so far?}
In particular, is it possible to devise comparably efficient division and modular
arithmetic operations favoring giant low structural complexity numbers? Would
that have an impact on primality and factoring algorithms?

The methodology to be used relies on two key components, that
have been proven successful so far, 
in discovering compact representations of important number-theoretical entities,
as well as low complexity algorithms for operations like  {\tt exp2},
{\tt add}, {\tt sub}, {\tt cmp}, {\tt mul} and {\tt bitsize}:
\begin{itemize}
\item partial evaluation  of functional programs with respect to known data types and operations on them,
 as well as the use of other program transformations
\item salient number-theoretical observations, similar 
to those in Props. \ref{fastexp}, \ref{bitcmp}, \ref{addeqs} \ref{subeqs},
\ref{bitineq}
that relate operations on our tree data types to 
number-theoretical identities  and algorithms.
\end{itemize}

Another aspect of future work is building a practical package
(that uses our representation only for numbers larger than the size of
the machine word) and specialize our algorithms 
for this hybrid representation. In particular,
parallelization of our algorithms, 
that seems natural given our tree representation,
would follow once the sequential performance
of the package is in a practical range.
Easier developments with practicality in mind would
involve extensions to signed integers and 
rational numbers.

\section*{Acknowledgement} This research has been supported
by NSF research grant 1018172.

\bibliographystyle{elsarticle-num}
\bibliography{INCLUDES/theory,tarau,INCLUDES/proglang,INCLUDES/biblio,INCLUDES/syn}

\section*{Appendix}

\subsection*{A subset of Haskell as an executable function notation}

We mention, for the benefit of the
reader unfamiliar with Haskell, that a notation like {\tt f x y} stands for $f(x,y)$,
{\tt [t]} represents sequences of type {\tt t} and a type declaration
like {\tt f :: s -> t -> u} stands for a function $f: s \times t \to u$
(modulo Haskell's ``currying'' operation, given the isomorphism between 
the function spaces ${s \times t} \to u$ and ${s \to t} \to u$). 
Our Haskell functions are always represented as sets
of recursive equations guided by pattern matching, conditional
to constraints (simple relations following \verb~|~ and before
the \verb~=~ symbol).
Locally scoped helper functions are defined in Haskell
after the {\tt where} keyword, using the same equational style.
The composition of functions {\tt f} and {\tt g} is denoted {\tt f . g}.
It is also customary in Haskell, when defining functions in an equational style (using {\tt =})
to write $f=g$ instead of $f~x=g~x$ (``point-free'' notation).
We also make some use of Haskell's ``call-by-need'' evaluation
that allows us to work with infinite
sequences, like the {\tt [0..]} infinite list notation, corresponding to the
set of natural numbers $\N$. Note also that the result
of the last evaluation is stored in the special Haskell
variable {\tt it}. By restricting ourselves to this {\em Haskell--}
subset, our code can also be easily transliterated into
a system of rewriting rules, other pattern-based functional
languages as well as deterministic Horn Clauses.

\subsection*{Division operations}
A fairly efficient integer division algorithm is  given here,
but it does not provide the same complexity gains as,
for instance, multiplication, addition or subtraction. Finding
a ``one block at a time'' division algorithm, if possible at all,
is subject of future work.
\begin{code}
div_and_rem x y | LT == cmp x y = (E,x)
div_and_rem x y | y /= E = (q,r) where 
  (qt,rm) = divstep x y
  (z,r) = div_and_rem rm y
  q = add (exp2 qt) z
    
  divstep n m = (q, sub n p) where
    q = try_to_double n m E
    p = leftshiftBy q m    
    
  try_to_double x y k = 
    if (LT==cmp x y) then s' k
    else try_to_double x (db y) (s k)   
          
divide n m = fst (div_and_rem n m)
remainder n m = snd (div_and_rem n m)
\end{code}
\subsection*{Integer square root}
A fairly efficient integer square root, using Newton's method is implemented as follows:
\begin{code}
isqrt E = E
isqrt n = if cmp (square k) n==GT then s' k else k where 
  two = i E
  k=iter n
  iter x = if cmp (absdif r x)  two == LT
    then r 
    else iter r where r = step x
  step x = divide (add x (divide n x)) two    
\end{code}

\begin{codeh}
-- experiments
 
collatz E = E
collatz x| o_ x = add x (o x)
collatz x| i_ x = hf x

ncollatz E = [E]
ncollatz x = x : ncollatz (collatz x)

c0 = take 1000  $(map (n.tsize) (ncollatz mersenne48))
c1 = take 100000  $(map (n.tsize) (ncollatz perfect48))
c2 = take 1000  $(map (n.tsize) (ncollatz prothPrime))
c3 = take 1000  $(map (n.tsize) (ncollatz cullenPrime))
c4 = take 1000  $(map (n.tsize) (ncollatz genFermatPrime))
c5 = take 1000  $(map (n.tsize) (ncollatz woodallPrime))
c6 = take 1000  $(map (n.tsize) (ncollatz sophieGermainPrime))

-- the syracuse sequence on large primes
t0 = take 1000  $(map (n.tsize) (nsyr mersenne48))
t1 = take 1000  $(map (n.tsize) (nsyr perfect48))
t2 = take 1000  $(map (n.tsize) (nsyr prothPrime))
t3 = take 1000  $(map (n.tsize) (nsyr cullenPrime))
t4 = take 1000  $(map (n.tsize) (nsyr genFermatPrime))
t5 = take 1000  $(map (n.tsize) (nsyr woodallPrime))
t6 = take 1000  $(map (n.tsize) (nsyr sophieGermainPrime))

tc = cmp (add cullenPrime  mersenne48) (add prothPrime mersenne48)

tf x = map (n.tsize) (nsyr (fermat (t x)))
tm x = map (n.tsize) (nsyr (mersenne (t x)))

assertOplus k x y =  (a==b) where
   a = add (otimes k x) (otimes k y)
   b = itimes k (add x y) 
   
assertIplus k x y =  (a==b) where
   a = add (itimes k x) (itimes k y)
   b = s' (s' (itimes k (s (s (add x y)))))
   
assertOiplus k x y =  (a==b) where
   a = add (otimes k x) (itimes k y) 
   b = s' (itimes k (s (add x y)))   

assertIoplus k x y =  (a==b) where
   a = add (itimes k x) (otimes k y) 
   b = s' (itimes k (s (add x y)))   
\end{codeh}

\end{document}

%% file: TOOLS/chheader13.tex
\usepackage{graphicx}
\usepackage{subfig} 
\usepackage{url}
\usepackage{verbatim}
\usepackage{color}
\definecolor{lgray}{gray}{0.92}
\definecolor{lblue}{rgb}{0.90,0.90,1.00}
\definecolor{lyellow}{rgb}{1.00,1.00,0.70}

\usepackage{listings}
\newenvironment{codex}{\small\verbatim}{\endverbatim\normalsize}

\lstnewenvironment{code}
    {\lstset{}%
      \csname lst@SetFirstLabel\endcsname}
    {\csname lst@SaveFirstLabel\endcsname}
\newtheorem{prop}{Proposition}

\newtheorem{df}{Definition}

\newcommand{\BI}[0]{\begin{itemize}}
\newcommand{\EI}[0]{\end{itemize}}

\newcommand{\BE}[0]{\begin{enumerate}}
\newcommand{\EE}[0]{\end{enumerate}}

\newcommand{\BX}[0]{\begin{codex}}
\newcommand{\EX}[0]{\end{codex}}
\newcommand{\BV}[0]{\begin{verbatim}}
\newcommand{\EV}[0]{\end{verbatim}}

\def \bscale1 {0.25}
\def \bscale {0.25}
\def \N {\mathbb{N}}
\def \T {\mathbb{T}}

\newcommand{\FIG}[4]{
\begin{figure}[htbp]
\centering
{\includegraphics[scale=#3]{figs/#4}}
\caption{#2}
\label{#1}
\end{figure}
}

